\documentclass{article}    %
\usepackage{arxiv}
\usepackage{bookmark}
\usepackage{cite}
\usepackage{amsmath}
\usepackage{lipsum} %
\usepackage{amssymb} %
\usepackage{mathrsfs} %
\usepackage{amsmath} %
\usepackage{amsthm} %
\usepackage{algorithm}  %
\usepackage{algpseudocode}%

\usepackage{graphicx}
\usepackage{tikz,pgfplots}
\usepackage{optidef}
\usepackage{multirow}
\usepackage{caption}
\usepackage{array}
\usepackage{adjustbox}
\usepackage{subcaption}
\usepackage{url}

\usepackage{fontawesome}
\newcommand{\cmarkOpt}{\faCheck}
\newcommand{\cmark}{\checkmark}

\usetikzlibrary{shapes,arrows, fit, calc}
\makeatother

\newtheorem{corr}{Corollary}
\newtheorem{lem}{Lemma}

\newtheorem{assum}{Assumption}
\newtheorem{defn}{Definition}
\newtheorem{prop}{Proposition}
\newtheorem{thm}{Theorem}
\newtheorem{prob}{Problem}

\newcommand{\ytask}{y_\mathrm{task}}
\newcommand{\dytask}{\dot{y}_\mathrm{task}}
\newcommand{\ktrust}{k_\mathrm{trust}}

\newcommand{\coln}{\mathrm{Col}}
\newcommand{\rown}{\mathrm{Row}}

\newcommand{\mcS}{\mathcal{S}}

\newcommand{\rowop}{\mathcal{R}}
\newcommand{\col}[1]{\rowop{\left(#1\right)}}
\newcommand{\rowsum}{+} %
\newcommand{\StateSpace}{\mathcal{X}} %
\newcommand{\StateSpaceXi}{\rowop\left(T_\mcS\right)} %
\newcommand{\StateSpaceEta}{\rowop\left(\TperpS\right)} %
\newcommand{\SOne}{ \mathcal{P}}
\newcommand{\STwo}{ \mathcal{Q}}
\newcommand{\SsetSubopt}{\SensorSet_\mathrm{subopt}}
\newcommand{\SSubopt}{\mcS_\mathrm{subopt}}
\newcommand{\colCtask}{\col{\Ctask}}
\newcommand{\colTsmalls}{\col{T_{s}}}
\newcommand{\colTsmallsOne}{\col{T_{s_1}}}
\newcommand{\colTsmallsTwo}{\col{T_{s_2}}}
\newcommand{\colTS}{\col{T_{\mcS}}}
\newcommand{\colTsOneOrsTwo}{\col{T_{\{s_1\} \cup \{s_2\}}}}
\newcommand{\colTStask}{\col{T_{\Stask}}}
\newcommand{\colTSOne}{\col{T_{\SOne}}}
\newcommand{\colTSTwo}{\col{T_{\STwo}}}
\newcommand{\colTSOneAndTwo}{\col{T_{\SOne \cap \STwo}}}
\newcommand{\colTSOneOrTwo}{\col{T_{\SOne \cup \STwo}}}

\newcommand{\colTSOrStask}{\col{T_{\mcS \cup \Stask}}}
\newcommand{\rank}{\mathrm{rank}} %
\newcommand{\Stask}{\mcS_{\mathrm{task}}} %
\newcommand{\StaskOne}{\mcS_{\mathrm{task,1}}} %
\newcommand{\StaskTwo}{\mcS_{\mathrm{task,2}}} %
\newcommand{\StaskThree}{\mcS_{\mathrm{task,3}}} %
\newcommand{\StaskFour}{\mcS_{\mathrm{task,1}}} %
\newcommand{\Ctask}{C_{\Stask}} %
\newcommand{\CtaskOne}{C_{\StaskOne}} %
\newcommand{\CtaskTwo}{C_{\StaskTwo}} %
\newcommand{\CtaskThree}{C_{\StaskThree}} %
\newcommand{\CtaskFour}{C_{\StaskFour}} %
\newcommand{\TperpS}{T_\mcS^\perp}

\newcommand{\dimR}[1]{\mathrm{dim}\left(#1\right)}
\newcommand{\SensorSet}{\mathscr{S}}
\newcommand{\PowerSet}{2^\SensorSet}
\newcommand{\Sgreedy}{\mcS^\ast_\mathrm{greedy}}
\newcommand{\SgreedyMinus}{\mcS^-_\mathrm{greedy}}
\newcommand{\situationawarenesstext}{sit\text{-}aware}
\newcommand{\Ssitaware}{{\SensorSet_\mathrm{\situationawarenesstext}}}
\newcommand{\Ssitawarered}{{\SensorSet_{\mathrm{\situationawarenesstext,reduced}}}}

\newcommand{\Sred}{{\SensorSet_\mathrm{reduced}}}

\newcommand{\Strust}{\SensorSet_\mathrm{trust}}

\newcommand{\Cref}{{C_{\mathrm{ref}}}}
\newcommand{\Sopt}{\mcS^\ast}

\makeatletter
\define@key{Gin}{Trim}
            {\let\Gin@viewport@code\Gin@trim\expandafter\Gread@parse@vp#1 \\}
\makeatother

\title{Trust-based user-interface design for human-automation systems
\thanks{This material is based upon work supported by the
    National Science Foundation.  Vinod, Thorpe, Olaniyi,
    and Oishi are supported under Grant Number CMMI-1254990, 
    CNS-1329878, and CMMI-1335038.  Summers
    is supported under Grant Number CNS-1566127 and
    CMMI-1728605 and by the Army Research Office under Grant Number W911NF-17-1-0058. Any opinions, findings, and conclusions or
    recommendations expressed in this material are those of
    the authors and do not necessarily reflect the views of
    the National Science Foundation.\newline 
    \indent Figure~\ref{fig:actual} is licensed from Dept.
    of Energy and Climate Change under
    Creative Commons Attribution-NoDerivs 2.0 Generic (CC
    BY-ND 2.0).  This photo appeared under the title ``Energy
    Minister Michael Fallon visits the National Grid Control
    Centre in Wokingham'' at
\texttt{https://www.flickr.com/photos/deccgovuk/8725424647/}
\texttt{in/photostream/}. 
}}
\author{Abraham~P.~Vinod\thanks{A. Vinod is with Oden Institute of Computational
    Engineering and Sciences, University of Texas at Austin,
    Austin, TX 78712 USA;
    e-mail:\texttt{aby.vinod@gmail.com}. This work was
completed while Vinod was a doctoral student at the
University of New Mexico.} , 
    Adam J. Thorpe$^\mathsection$, 
    Philip~A.~Olaniyi$^\mathsection$, 
    Tyler~H.~Summers\thanks{T. Summers is with Mechanical
    Engineering, University of Texas at Dallas, Richardson,
TX 75080 USA; e-mail: \texttt{tyler.summers@utdallas.edu}} , 
and Meeko~M.~K.~Oishi\thanks{A. Thorpe, P. Olaniyi, and M.
    Oishi (corresponding author) are with Electrical and
    Computer Engineering, University of New Mexico,
    Albuquerque, NM 87131 USA;
e-mail:\texttt{\{ajthor,polaniyi,oishi\}@unm.edu}}%
}
\date{}
\renewcommand{\headeright}{}    
\renewcommand{\undertitle}{}    

\begin{document}
\maketitle

\begin{abstract}
We present a method for dynamics-driven, user-interface design for a human-automation system via sensor selection.  We define the user-interface to be the output of a MIMO LTI system, and formulate the design problem as one of selecting an output matrix from a given set of candidate output matrices.  Sufficient conditions for situation awareness are captured as additional constraints on the selection of the output matrix.  These constraints depend upon the level of trust the human has in the automation.  We show that the resulting user-interface design problem is a combinatorial, set-cardinality minimization problem with set function constraints.  We propose tractable algorithms to compute optimal or sub-optimal solutions with suboptimality bounds. Our approaches exploit monotonicity and submodularity present in the design problem, and rely on constraint programming and submodular maximization. We apply this method to the IEEE 118-bus, to construct correct-by-design interfaces under various operating scenarios. 
\end{abstract}

\keywords{
    User-interface design\and human-automation interaction\and
    observability\and sensor selection\and output synthesis
}

\section{Introduction}

{\em Situation awareness}, the ability to deduce the current state of the system and predict the evolution of the state in the short-term \cite{Endsley1995}, 
is essential for effective human-automation interaction.  
In expensive, high-risk, and safety-critical systems, such as
power grid distribution systems, aircraft and other transportation systems, 
biomedical devices, and nuclear power
generation, the user-interface %
helps the user maintain situation awareness by 
providing critical information about the
system to the user \cite{sheridan1992telerobotics,Dix91}.
Indeed, a lack of situation awareness is known to be a contributing factor to operator error in major grid failures \cite{Panteli2013, Endsley2008}.
A variety of recommendations and guidelines for ``good'' user-interface design have been posited \cite{steinfeld2004interface,PF11,Billings97}.
However, {\em formal tools} for user-interface
design, that explicitly incorporate 
the underlying dynamics, could help
avert potential errors and mishaps, and reduce time consuming and costly
design and testing iterations.

We consider the user-interface to be equivalent to an output
map of the dynamical system, and pose the question of
user-interface design as one of \emph{sensor selection}:
among the sensors that could be the elements of the
interface, we aim to identify a combination which is minimal \cite{Billings97}, yet sufficient for 
situation awareness, and dependent upon the user's trust in the automation. 
We focus solely on the information content, and not on the qualitative aspects of {\em how} that information is provided.  The need for minimal interfaces is particularly evident in large systems (Figure \ref{fig:actual}), for which providing too much 
information can render the interface ineffective because it is overwhelming, 
and providing too little information can result in perceived non-determinism \cite{DH02}.

\begin{figure}
    \centering
    \includegraphics[width=0.5\linewidth]{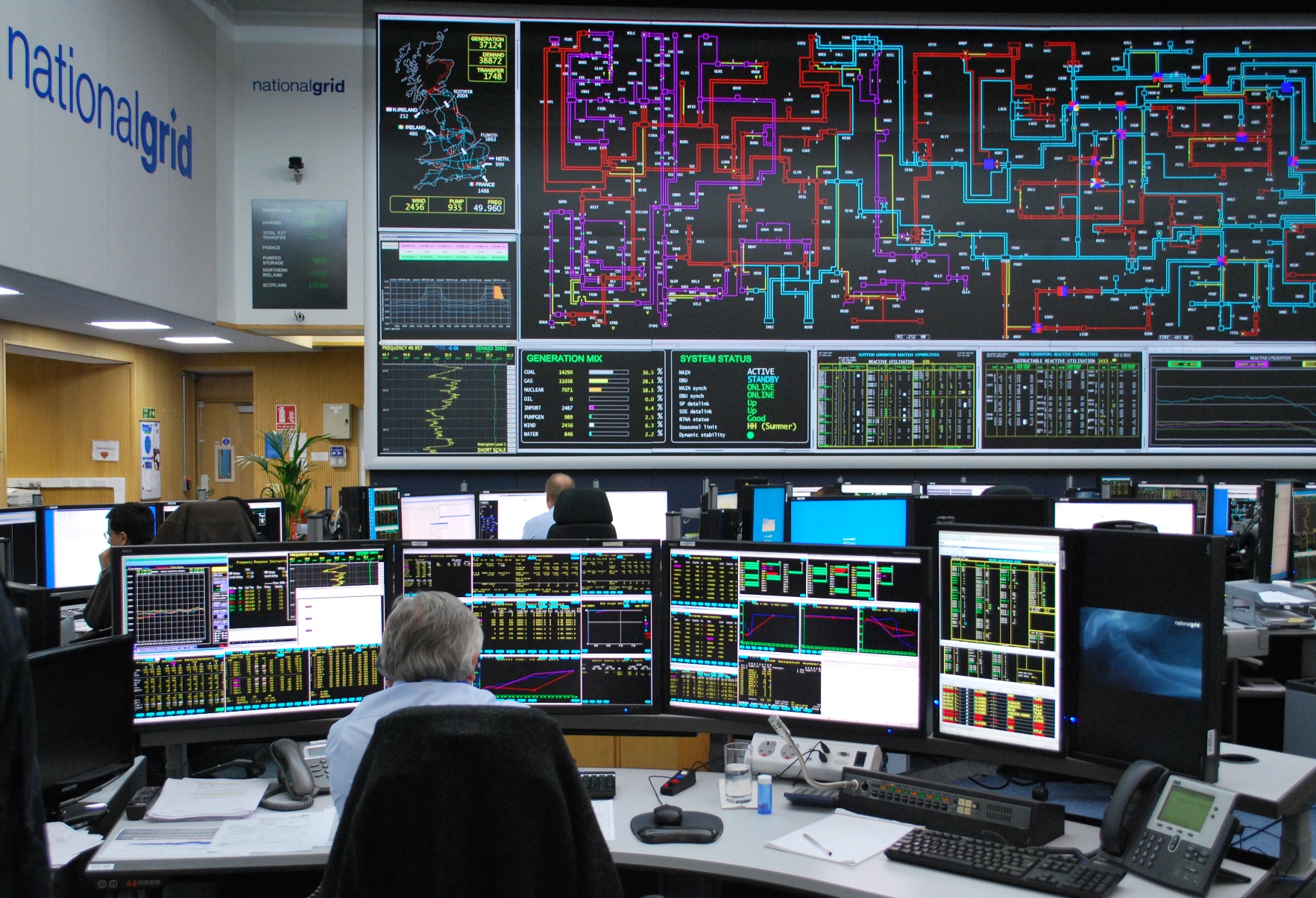}
    \caption{User-interfaces for power grid operators facilitate
        situation awareness, by providing
        information from which the power grid operator can estimate the state 
        and predict its evolution.  The sheer volume of
        information warrants the use of constructive tools
        (as opposed to ad-hoc guidelines) to synthesize the
        information content of the interface. Image licensed
    under CC BY-ND 2.0.}
    \label{fig:actual}
\end{figure}

Sensor selection \cite{VT_RSS_2010, mourikis2006optimal,rowaihy2007survey,qi2014optimal, gou2008generalized} is typically posed as a combinatorial optimization
problem, %
which becomes intractable even for moderate problem
sizes.  While some heuristics, such as convex
relaxation~\cite{JoshiConvex, polyak2013lmi} and 
combinatorial algorithms that avoid a full exhaustive search
\cite{krause2007near, shamaiah2010greedy, summers_2014} have been employed, computational complexity remains a significant challenge.  
For some problem classes (e.g., cardinality-constrained
submodular set function maximization~\cite{krause2007near,
nemhauser_1978}), greedy algorithms and other graph
theoretic approaches can yield provably optimal or
near-optimal results
\cite{clark_submodularity_2017,wolsey1982analysis,
shamaiah2010greedy, krause2007near, nemhauser_1978,
summers_2014, berger2005discrete}.  
Hence we focus heavily in this paper on characterization of the computational aspects of user-interface design via sensor selection.

Other approaches to user-interface analysis and design have focused on related aspects of human-automation interaction.  Model-checking has been used to detect mode confusion 
in discrete
event systems \cite{Rushby14, PF11, BBS13}, and finite-state
machine reduction techniques have been
used to synthesize user-interfaces of minimal cardinality
for discrete-state abstractions of hybrid systems \cite{DH02, OMBDT08}.  
Interfaces have been designed to assure internal and external awareness \cite{Sadigh2017}, 
to facilitate transfer of control authority between the human and the automation, 
and to articulate information related to the role of regret in human decision aids \cite{Wang2018}.  
In \cite{Jain2018_1,Jain2018_2}, the effect of transparency on workload and trust was evaluated, and a feedback scheme developed that alters transparency of the interface.  
Other interface design approaches focus on moderating human input \cite{BLM12,Murphey2016} despite uncertainty, and on mixed-initiative control \cite{Sadrfaridpour2017,Saeidi2017} for human-robot interaction.

Our approach is based on observability conditions that presume
the human is a special type of observer, %
to assess whether the interface provides sufficient information for the
human to accomplish a given task \cite{EO11, oishi_2014, Hammond15}.
Hence in contrast to standard sensor placement problems, additional 
constraints arise to ensure situation awareness, and to capture the effect of the user's trust in the automation.
{\em The main contributions of this paper are: 1) assurances of optimality and suboptimality via submodularity and monotonicity properties, specific to the user-interface design problem, and 2) efficient numerical implementations, that employ constraint programming, greedy heuristics for submodular
maximization, and a novel enumeration framework for large user-interface design problems.}
The algorithmic advances proposed here enable application to problems that would be computationally prohibitive with our preliminary approach \cite{VinodACC2016}.
Further, the model proposed here captures gradated user trust in the automation, a more 
subtle characterization than the simplistic, no trust or full trust, characterization
that was used in \cite{VinodACC2016}.

The paper is organized as follows: Section~\ref{sec:prem}
provides the problem formation.  Section~\ref{sec:UI}
formulates user-interface design as a combinatorial
optimization problem.  Section~\ref{sec:CP} describes a
novel enumeration framework that enables computationally
efficient search for feasible user interfaces.
Section~\ref{sec:app} demonstrates our approach on
user-interface design for a large system, the IEEE 118-bus,
and Section~\ref{sec:conc} provides the conclusions.

\section{Preliminaries and problem statement}
\label{sec:prem}

A finite set $\SensorSet$ has cardinality
$| \SensorSet|$ and power set $\PowerSet$. 
A set function $f: \PowerSet \rightarrow \mathbb{R}$ takes
as input a subset of $\SensorSet$ and returns a real number.
For natural numbers
$a,b\in \mathbb{N}$ with $a\leq b$, we 
define the set $\mathbb{N}_{[a,b]}=\{c\in
\mathbb{N}:a\leq c\leq b\}$. For a matrix $M\in
\mathbb{R}^{p\times q}$, we denote its column rank by
$\rank(M)$, and its column space (range) by $\col{M}$. We define a
matrix whose column space coincides with a subspace $
\mathcal{V}$ as $\mathrm{basis}(\mathcal{V})$. Recall that
$\mathrm{basis}\left(\col{M}\right)$ is not unique.  Given
two vector spaces $ \mathcal{V}_1,\mathcal{V}_2$, their sum
($\mathcal{V}_1\rowsum\mathcal{V}_2=\{v_1+v_2:v_1\in
\mathcal{V}_1,v_2\in \mathcal{V}_2\}$) and their
intersection  are vector spaces~\cite[Pg.
22]{friedbergLinearAlg}.

\begin{figure}
    \centering
    \includegraphics[width=0.5\linewidth]{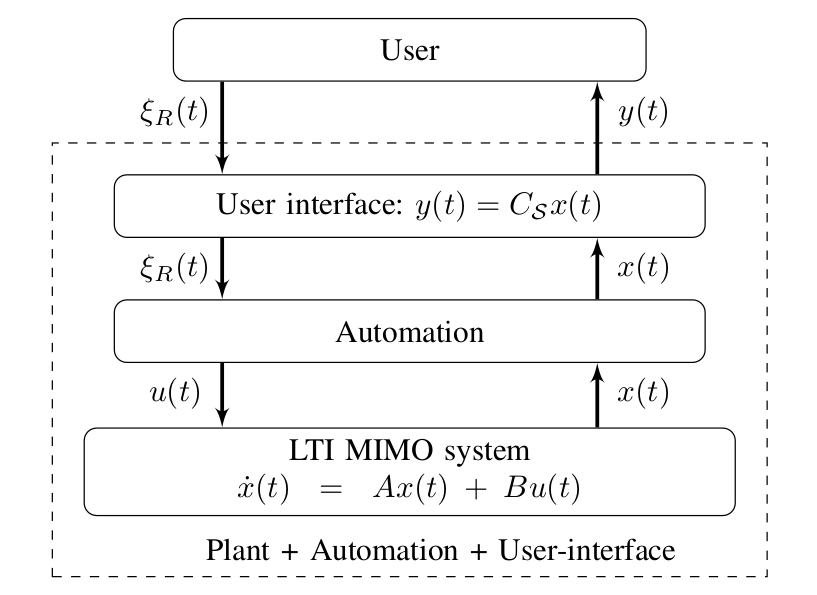} 
    \caption{Human-automation system in which the human
    provides a reference trajectory, and the automation
synthesizes a low-level control to achieve it.}
    \label{fig:manualcontrol}
\end{figure}

Consider a human-automation
system %
(Figure~\ref{fig:manualcontrol}) in which the human provides a reference trajectory  $\xi_R(t) \in \mathbb R^p$, and the automation synthesizes a low-level controller to achieve reference tracking \cite{oishi_2014}.  We presume a MIMO LTI system, 
\begin{subequations}
\begin{align}
    \dot{x}(t) &= A x(t) + B u(t)  \label{eq:sys_dyn}\\
    y(t)&= C_\mcS x(t)\label{eq:sys_output}
\end{align}\label{eq:sys}%
\end{subequations}%
with state $x(t)\in \StateSpace = \mathbb{R}^n$, input
$u(t)\in \mathbb{R}^m$, output $y(t)\in \mathbb{R}^p$, and
known matrices $A\in \mathbb{R}^{n\times n}$ and $B\in
\mathbb{R}^{n\times m}$.  
The user receives
information about the plant via the user-interface. 

\begin{defn}[\textbf{User-interface and sensors}]\label{defn:UIsensor}
    We define the output $y(t)$ as the user-interface of the
    system \eqref{eq:sys}, with the candidate rows of
    $C_\mcS$ referred to as the sensors $s_i \in
    \mathbb{R}^{n}$. 
\end{defn}
By Definition~\ref{defn:UIsensor}, a sensor is a potential
element of the user-interface.  We denote the set of all
sensors as
$\SensorSet = \{s_1, \cdots, s_{| \SensorSet |}\}$ for a
finite ${| \SensorSet |}\in \mathbb{N}$.  For any sensor
combination $\mcS\in\PowerSet$, the output matrix $C_\mcS$
is a matrix whose rows consist of the elements $s_i\in\mcS$,
and the total number of outputs associated with $C_\mcS$ is
$p=| \mcS|$.

\begin{defn}[\textbf{Task}]\label{defn:task}
    A task is characterized by the tuple $(\ell,\Ctask)$,
    with a known task matrix $\Ctask \in \mathbb{R}^{{|
    {\Stask}|}\times n}$ associated with $\Stask\in
    \PowerSet$, and a known, possibly nonlinear, function
    $\ell: \mathbb{R}^{|\Stask|} \rightarrow \mathbb{R}$.
    The task $(\ell,\Ctask)$ is a specification of the form
    always $x(t)\in \mathcal{F}(t)$ or eventually $x(t)\in
    \mathcal{F}(t)$, for $ \mathcal{F}(t) =
    \{x(t):\ell\left(\Ctask x(t)\right) \geq
    0\}$.
\end{defn}
The task is defined in terms of safety or
liveness specifications, i.e., a desirable phenomena that
should always or eventually
happen~\cite{oishi_2014,baier2008principles}.  The task may
also be interpreted as imposing a specification on the
output $\ytask(t)=\Ctask x(t)$.  

\emph{Illustrative example}:
Consider an LTI model of a jerk-controlled robot constrained
to move in a line, which is tasked with maintaining a
velocity above a minimum speed $v_\mathrm{min}$.
The robot has a camera mounted with independent dynamics.
The position dynamics (3D) and camera heading dynamics (1D)
result in
\begin{align}
    A&= \left[ {\begin{array}{cccc} 
                        0  &  1 & 0 & 0 \\ 
                        0  &  0 & 1 & 0 \\ 
                        0  &  0 & 0 & 0 \\
                        0  &  0 & 0 & 0 \\
                   \end{array} } \right], \: %
    B = \left[ {\begin{array}{cc} 
                        0 & 0 \\ 
                        0 & 0 \\
                        1 & 0 \\
                        0 & 1 \\
                \end{array} } \right]\label{eq:sys_chain}
\end{align}
with states that include position, velocity, acceleration,
and the camera heading.  We consider a suite of sensors
based on measurements of each state, i.e.,
$ \SensorSet = \{s_p, s_v, s_a, s_h\}$ with
$s_p=[1\ 0\ 0\ 0]$, $s_v=[0\ 1\ 0\ 0]$, $s_a=[0\ 0\ 1\ 0]$,
and $s_h=[0\ 0\ 0\ 1]$. 
The task is defined by
$(\ell,\Ctask)$, with $\Stask = \{s_v\}$, $\Ctask =
s_v$, $\ell(z)=z-v_\mathrm{min}$, and $
\mathcal{F}(t)=\{x(t): \Ctask x(t)\geq v_\mathrm{min}\}$.
\hfill $\qedsymbol$
\vspace*{5pt}

For a given task $(\ell, \Ctask)$, we seek to design a user-interface 
$C_\mathcal{S}$ that satisfies, in order of
importance:
\begin{enumerate}
    \item[C1)] Situation awareness,
    \item[C2)] Compatibility with the user's trust in the
        automation,~and
    \item[C3)] Conciseness.
\end{enumerate}
These properties represent human factors that are key for effective human-automation interaction, and will be described in detail in Section \ref{sec:UI}.  Briefly, constraint C1 takes into account the limitations of the
human operator and the complexity of the task.  
Constraint C2 requires that more information is provided to the
user when the user's trust in the automation is low, and vice versa. Constraint
C3 prevents high cognitive load
associated with excessive data.  

We embed these properties as constraints in the sensor selection problem for user-interface design:
\begin{subequations}
    \begin{align}
        \underset{\mcS\in\PowerSet}{\mathrm{minimize}}&\quad
        | \mcS| &\mbox{(concise)}\label{eq:prob_defn_cost}\\
        \mathrm{subject\ to}&\quad
        \mcS\in\Ssitaware\label{eq:prob_defn_constraint2}
        &\mbox{(situation awareness)}\\
                            &\quad
        \mcS\in\Strust\label{eq:prob_defn_constraint3}
        &\mbox{(trust)}
    \end{align}\label{prob:prob_defn}%
\end{subequations}%
in which \eqref{eq:prob_defn_cost} arises from C3, 
\eqref{eq:prob_defn_constraint2} arises from C2, and 
\eqref{eq:prob_defn_constraint3} arises from C1.

\begin{prob}\label{prob:main}
    Given a task $(\ell,\Ctask)$ and a human-automation
    system \eqref{eq:sys}, find a succinct characterization
    of the constraint for situation awareness, $\Ssitaware$, and of the
    constraint for trust compatibility, $\Strust$.
\end{prob}

\begin{prob}
    Construct tractable combinatorial optimization
    algorithms to solve \eqref{prob:prob_defn}, with
    guarantees of optimality or suboptimality, as appropriate.
\end{prob}

Because combinatorial optimization problems are typically hard to
solve due to their large feasible solution space, solving (\ref{prob:prob_defn}) directly is a challenging endeavour.  Problem 1 provides structure that we can exploit
to address Problem 2, so that (\ref{prob:prob_defn}) can be addressed through tractable reformulation.

\section{User-interface design as sensor selection}
\label{sec:UI}

\subsection{Situation awareness via observability}
\label{sub:sit_aware}

Situation awareness consists of three elements: perception, comprehension, and projection, more formally defined in \cite{Endsley1995} as ``perception of the elements in an environment within a volume of time and space, the comprehension of their meaning, and the projection of their status in the near future.''
As in \cite{EO11,Hammond15,Eskandari2016}, we interpret these three elements respectively as
as the ability to reconstruct those elements of the state
that are relevant for the task at hand, the ability to understand the output and its time derivatives, and the ability to reconstruct those elements of the state derivative relevant to the task at hand.  Unlike standard notions of observability, we do not require that the entire state can be reconstructed \cite{EO11,oishi_2014}.  

Although qualitative aspects of user-interface design are key for effective human-automation interaction \cite{Dix91,Oulasvirta2020}, we focus solely on quantitative aspects, and presume that information content will be presented in a human-centric manner.  In essence, 
we presume that a user with situation awareness is able to make sense of the presented information for the purpose of decision-making and control \cite{Endsley1995,wickens2008situation}.  

\begin{assum}[\textbf{Situation awareness}]\label{assum:SitAware}
    For a given a user-interface, 
    constructed from elements $\mcS\in\PowerSet$, 
    a user with situation awareness can reconstruct the output of
    the system, $y(t)= C_{\mcS} x(t)$, the unforced
    higher derivatives of the output, and their linear
    combinations.
\end{assum}

As in \cite{EO11,Hammond15,VinodACC2016}, we %
employ input-output linearization 
to capture the user's
interaction with the system (\ref{eq:sys}).  
We presume the user provides a reference trajectory $\xi_R(t)$ that is smooth.  
Given an output matrix $C_\mcS$ with 
$\mcS=\{s_1,s_2,\ldots,s_{| \mcS |}\}\in \PowerSet$, we
construct
a similarity
transform $P_\mcS \in \mathbb{R}^{n\times n}$,
\begin{equation}
\left[\begin{array}{c}
        \xi(t)\\
        \hline\eta(t)
      \end{array}\right]
   =P_{\mcS}x(t)
   =\left[\begin{array}{c}
           T_\mcS\\
           \hline\TperpS
          \end{array}\right]x(t)\label{eq:MIMO_transform}
\end{equation}
that results in observable states $\xi(t)\in \StateSpaceXi$ and unobservable states $\eta(t)\in
\StateSpaceEta$.  The linear transformation $T_\mcS$ is
defined using $T_{s_i}$ for some $s_i\in\mcS$, as
\begin{align}
    T_{s_i}&={\left[s_i\ {(s_i^\top A)}^\top\ {(s_i^\top
    A^2)}^\top\ \ldots\ {(s_i^\top
    A^{\gamma(s_i)-1})}^\top\right]}^\top\label{eq:Tsi},\\
    T_\mcS &= \mathrm{basis}\left({\col{\left[{T_{s_1}^\top\
    T_{s_2}^\top\ \ldots\ T_{s_{| \mcS
    |}}^\top}\right]}^\top}\right) \label{eq:TS},
\end{align}
where $\gamma: \SensorSet \rightarrow \mathbb{N}_{[1,n]}$ is
the relative degree of the MISO system with the single
output $s_i^\top x(t)$.  By \eqref{eq:Tsi}, $\col{T_{\mcS}}$
is the state subspace spanned by the outputs characterized
by $y(t)=C_{\mcS}x(t)$ and their unforced  higher
derivatives.

\begin{assum}\label{assum:AutoWellDesigned}\textbf{(Correctly designed automation)}
    The automation generates $u(t)$ such that that $(\xi(t),
    \dot{\xi}(t))$ tracks the reference trajectory
    $(\xi_R(t), \dot{\xi}_R(t))$.
\end{assum}

The implications of Assumptions 1 and 2 are twofold: 1) the user can reconstruct $\xi(t)$
and predict its evolution (because $\dot{\xi}(t)$ can be reconstructed), and 
2) the user
delegates control of the internal dynamics $\eta(t)$ to
the automation.

To tractably 
enumerate $\Ssitaware$, we propose the 
\emph{user information index}, a set function %
that measures the dimension of the state subspace the
user can reconstruct and predict from the information
presented in the user interface.  

\begin{table}
    \scriptsize
    \centering
    \setlength\tabcolsep{1pt}
    \newcommand{\propWidth}{2.2cm}                
    \newcommand{\oneSenseWidth}{0.6cm}                
    \newcommand{\twoSenseWidth}{1.0cm}
    \newcommand{\threeSenseWidth}{0.9cm}
    \begin{tabular}{|m{\propWidth}|m{\oneSenseWidth}|m{\oneSenseWidth}|m{\oneSenseWidth}|m{\oneSenseWidth}|m{\twoSenseWidth}|m{\twoSenseWidth}|m{\twoSenseWidth}|m{\twoSenseWidth}|m{\twoSenseWidth}|m{\twoSenseWidth}|m{\threeSenseWidth}|m{\threeSenseWidth}|m{\threeSenseWidth}|m{\threeSenseWidth}|m{\twoSenseWidth}|}
    \hline									%
    $\mcS$ & $\{s_p\}$ & $\{s_v\}$ & $\{s_a\}$ & $\{s_h\}$ & $\{s_p,s_v\}$ &
    $\{s_p,s_a\}$ & $\{s_p,s_h\}$ & $\{s_v,s_a\}$ & $\{s_v,s_h\}$ &
    $\{s_a,s_h\}$ &  \begin{minipage}{\twoSenseWidth}
        $\{s_p,s_v,$\\\hspace*{4pt} $s_a\}$
    \end{minipage} & \begin{minipage}{\twoSenseWidth}
        $\{s_p,s_v,$\\\hspace*{4pt} $s_h\}$
    \end{minipage} & \begin{minipage}{\twoSenseWidth}
        $\{s_p,s_a,$\\\hspace*{4pt} $s_h\}$
    \end{minipage} & \begin{minipage}{\twoSenseWidth}
        $\{s_v,s_a,$\\\hspace*{4pt} $s_h\}$
    \end{minipage} & \begin{minipage}{\twoSenseWidth}
        $\{s_p,s_v,$\\\hspace*{2pt} $s_a,s_h\}$
    \end{minipage}\\ \hline
    $\Gamma(\mcS)$              & 3 & 2& 1& 1& 3& 3& 4& 2& 3& 2& 3& 4& 4& 3& 4\\\hline
    $\Gamma(\mcS\cup\Stask)$    & 3 & 2& 2& 3& 3& 3& 4& 2& 3& 3& 3& 4& 4& 3& 4 \\\hline 
    $\Ssitaware$                   & \cmark& \cmark&&&\cmark&\cmark&\cmark&\cmark&\cmark&&  \cmark&\cmark&\cmark&\cmark&\cmark\\\hline
    $2^\Sred$                  & \cmark& \cmark& \cmark & &\cmark&\cmark& &\cmark& & &\cmark & & & & \\\hline
    $\Ssitawarered$                & \cmark& \cmark& & &\cmark&\cmark& &\cmark& & &\cmark & & & & \\\hline
    $\Strust,\ktrust=1$         & \cmarkOpt & \cmarkOpt & \cmark & \cmark & \cmark &
    \cmark & \cmark & \cmark & \cmark & \cmark & \cmark & \cmark & \cmark &
    \cmark & \cmark \\\hline
    $\Strust,\ktrust=2$         & \cmarkOpt & \cmarkOpt &  &  & \cmark &
    \cmark & \cmark & \cmark & \cmark & \cmark & \cmark & \cmark & \cmark &
    \cmark & \cmark \\\hline
    $\Strust,\ktrust=3$         & \cmarkOpt & & & & \cmark &
    \cmark & \cmark &  & \cmark & & \cmark & \cmark & \cmark &
    \cmark & \cmark \\\hline    
    $\Strust,\ktrust=4$         & & & & &  & & \cmarkOpt &  &  &  &  &\cmark   &  \cmark & & \cmark \\\hline    
    \end{tabular}
    \caption{Application of various definitions to the illustrative example
    given in Section~\ref{sec:prem}. Here, $\SensorSet=\{s_p,s_v,s_a,s_h\}$,
$\Stask=\{s_v\}$, and $\Sred=\{s_p,s_v,s_a\}$.  Interfaces that satisfy both situation awareness and trust constraints for a given level of user-trust in the automation are feasible for \eqref{prob:prob_defn}; interfaces that are optimal for a given trust level are indicated in bold. } \label{tab:table_example}
\end{table}

\begin{figure}[ht]
    \centering
    \includegraphics[width=0.6\linewidth]{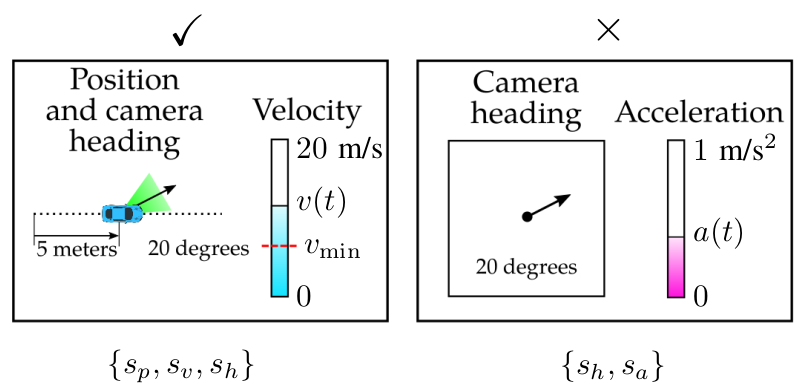} 
    \caption{User-interfaces for the illustrative
    example in Section~\ref{sec:prem}. The
user-interface on the left enables situation awareness for
the task of keeping $v(t) \geq v_\mathrm{min}$, and is appropriate for all levels of user trust.  In contrast, the
user-interface on the right does not enable situation awareness, and meets the trust requirement only for high levels of trust.}
    \label{fig:table_example_UIs}
\end{figure}

\begin{defn}[\textbf{User information index}]\label{defn:Uii}
    The \emph{user information index} is the set
    function $\Gamma:\PowerSet \rightarrow
    \mathbb{N}_{[1,n]}$,
    \begin{align}
        \Gamma(\mcS)&=\dimR{\col{T_\mcS}}=\rank( T_\mcS). 
        \label{eq:gamma_def}
    \end{align}
\end{defn}
The user information index $\Gamma(\mcS)$ characterizes the
dimensions of $\xi(t)$ and $\eta(t)$, since
$\xi(t)\in\mathbb{R}^{\Gamma(\mcS)}$ and
$\eta(t)\in\mathbb{R}^{n-\Gamma(\mcS)}$. 
Table~\ref{tab:table_example} shows $\Gamma(\mcS)$ for the illustrative example.

\begin{prop}[\textbf{Sufficient information for task completion}]\label{prop:SitAware}
    If $\colCtask\subseteq\colTS$, %
    then the user-interface $C_{\mcS}$ provides sufficient
    information to complete the task $(\ell,\Ctask)$. 
\end{prop}
\begin{proof}
    If $\colCtask\subseteq\colTS$, we can express the task
    output $\ytask(t)=\Ctask x(t)\in\colCtask$ as a linear
    combination of the observable state $\xi(t)\in\colTS$.  Hence under
    Assumptions~\ref{assum:SitAware}
    and~\ref{assum:AutoWellDesigned}, the user can estimate
    $\ytask(t)$ and $\dytask(t)$ from the user-interface
    output $y(t)=C_{\mcS}x(t)$. 
\end{proof}

Proposition~\ref{prop:SitAware} states that a user-interface
enables situation awareness of the task at hand, provided $\ytask(t)$ is
contained in the observable subspace $\colTS$. 
However, the conditions in Proposition 1 are not amenable to tractable computation. Hence, we reframe Proposition \ref{prop:SitAware} in terms of the user-information index.

\begin{lem} %
    \label{lem:prop_TS_gamma}
   Given any $\SOne,\STwo\in\PowerSet$,
   {\renewcommand{\theenumi}{\alph{enumi}}
   \begin{enumerate}
       \item $\SOne\subseteq \STwo$ implies $\colTSOne\subseteq \colTSTwo$ and $\Gamma(\SOne)\leq\Gamma(\STwo)$,\label{lem:prop_TS_gamma_row_TS_monotone}
       \item $\Gamma(\SOne\cup \STwo) =\Gamma(\SOne) + \Gamma(\STwo) - \dimR{\colTSOne\cap \colTSTwo}$,\label{lem:prop_TS_gamma_cup_row}
       \item $\Gamma(\SOne\cap \STwo)\leq\dimR{\colTSOne\cap
           \colTSTwo}$, and\label{lem:prop_TS_gamma_cap}
       \item $\Gamma(\SOne\cup \STwo) =\Gamma(\SOne)$ if and
           only if $\colTSOneOrTwo=\colTSOne$.\label{lem:prop_TS_gamma_TS_equality_converse}
    \end{enumerate}}
\end{lem}
The proof of Lemma \ref{lem:prop_TS_gamma} is provided in Appendix \ref{app:lem_prop_TS_gamma}.

\newcommand{\proofLemPropTSGamma}{
    We have \ref{lem:prop_TS_gamma_row_TS_monotone}) from \eqref{eq:Tsi}, \eqref{eq:TS}, and \eqref{eq:gamma_def}.

    We have \ref{lem:prop_TS_gamma_cup_row}) by
    \eqref{eq:TS_union} and~\cite[Sec. 1.6, Ex.
    29]{friedbergLinearAlg},
    \begin{align}
        \colTSOneOrTwo&=\colTSOne\rowsum \colTSTwo.\label{eq:TS_union}
    \end{align}

    We have \ref{lem:prop_TS_gamma_cap}) from
    \eqref{eq:TS_intersect}~\cite[Sec. 1.4, Ex.
    15]{friedbergLinearAlg} and
    \ref{lem:prop_TS_gamma_row_TS_monotone}),
    \begin{align}
        \colTSOneAndTwo \subseteq \colTSOne \cap
        \colTSTwo\label{eq:TS_intersect}.
    \end{align}

    We have \ref{lem:prop_TS_gamma_TS_equality_converse}) 
    from \ref{lem:prop_TS_gamma_row_TS_monotone}) and
    \ref{lem:prop_TS_gamma_cup_row}).
}

\begin{prop}[\textbf{Situation awareness constraint via
user information index}]\label{prop:Ssitaware_defn}
For every $\mcS\in\Ssitaware$, defined as
\begin{align}
    \Ssitaware&\triangleq\left\{{\mcS\in \PowerSet:
    \Gamma(\mcS)=\Gamma(\mcS\cup\Stask)}\right\}
    \label{eq:Ssitaware_third},
\end{align}
the user interface $C_\mcS$
provides sufficient information to complete the task.
\end{prop}
\begin{proof}
    By \eqref{eq:TS}, $\colCtask\subseteq\colTStask$.  By Lemma~\ref{lem:prop_TS_gamma}\ref{lem:prop_TS_gamma_TS_equality_converse},
     we have
        $\Ssitaware=\left\{{\mcS\in \PowerSet: \colTS=\colTSOrStask}\right\}$. %
    Further,
    $\colTStask\subseteq\colTSOrStask$ for any $\mcS\in\Ssitaware$ by
    Lemma~\ref{lem:prop_TS_gamma}\ref{lem:prop_TS_gamma_row_TS_monotone}.
    Hence we have 
    $\colTStask\subseteq\colTSOrStask=\colTS$ for any $\mcS\in\Ssitaware$.  Thus,
    $\colCtask\subseteq\colTS$. 
    Applying Proposition~\ref{prop:SitAware} completes the proof.
\end{proof}

Table~\ref{tab:table_example} shows $\Ssitaware$ for the
illustrative example, and two possible interfaces are shown in Figure \ref{fig:table_example_UIs}. As expected,
$\{s_h\}\not\in\Ssitaware$ since the heading measurement
$s_h$ alone provides no information about velocity (the
task), due to the decoupled dynamics \eqref{eq:sys_chain}.
Furthermore, $\{s_a\}\not\in\Ssitaware$ since reconstructing
velocity from acceleration measurements requires
integration.
Thus, all
sensor combinations in $\PowerSet\setminus\{\{s_a\}
,\{s_h\},\{s_a,s_h\}\}$ provide sufficient information for situation  awareness, enabling task completion.

\begin{algorithm}[t]
    \caption{Efficient enumeration of $\Ssitaware$ via characterization of $\Ssitawarered$
}\label{algo:Enum}
    \begin{algorithmic}[1]
        \Require~Set of all sensors $\SensorSet$, sensors that define the task
        $\Stask$, the user information index
        $\Gamma(\cdot)$ 
        \Ensure~Sensor combinations that enable situation awareness
         $\Ssitaware$, and a reduced set $\Ssitawarered$
        \State $\Ssitaware\gets \emptyset,\quad\Ssitawarered\gets \emptyset$
        \State Compute $\Ssitawarered$ using \eqref{eq:soneset}\label{line:soneset}
        \For{$\SOne\in\Ssitawarered$}
                \State $\Ssitaware \gets \Ssitaware \cup \{\SOne\times
                2^{\SensorSet\setminus\SOne}\}$\label{line:Senum}
        \EndFor
        \State\Return $(\Ssitaware,\Ssitawarered)$
  \end{algorithmic}
\end{algorithm}

Since enumerating $\PowerSet$ to compute $\Ssitaware$ 
is computationally expensive for large $|\SensorSet|$, 
we propose Algorithm~\ref{algo:Enum} for a tractable
enumeration of $\Ssitaware$. 
We construct a reduced set of admissible sensors $\Sred$,
\begin{align}
    \Sred&\triangleq\{s\in \SensorSet:
            \Gamma(s)+\Gamma(\Stask)>\Gamma(s\cup\Stask)\}\label{eq:Sred}\\
         &=\{s\in \SensorSet:
            \dimR{\colTsmalls\cap \colTStask}>0\},\label{eq:Sred_dim}
\end{align}
(where \eqref{eq:Sred_dim} follows from \eqref{eq:Sred} and
Lemma~\ref{lem:prop_TS_gamma}\ref{lem:prop_TS_gamma_cup_row}), 
to construct an easily computable subset of $\Ssitaware$,
    \begin{align}
        \Ssitawarered&=
        \{\SOne\in2^{\Sred} |
        \Gamma(\SOne\cup\Stask)=\Gamma(\SOne)\}.
        \label{eq:soneset}
    \end{align}
The set $\Ssitawarered$ is ``minimal,'' in that removing any sensor from the sensor combinations in $\Ssitawarered$ will violate the situation awareness constraint \eqref{eq:Ssitaware_third}. 
Additional elements are appended to $\Ssitawarered$ (line~\ref{line:Senum}), so that 
Algorithm~\ref{algo:Enum} provides an exact enumeration of
the members of $\Ssitaware$.  Algorithm~\ref{algo:Enum} is computationally tractable, since enumeration is done over $2^\Sred$, and 
$|2^\Sred| << |\PowerSet|$.

\begin{thm}[\textbf{Correctness of
    Algorithm~\ref{algo:Enum}}]\label{thm:Ssitaware_decompose}
    The set $\Ssitaware$ can be constructed as the union of two sets,
    \begin{align}
        \Ssitaware &= \left\{\mcS\in\SensorSet\middle|
            \begin{array}{c}
                \SOne=\mcS\cap\Sred,\\
                \Gamma(\SOne\cup\Stask)=\Gamma(\SOne)
            \end{array}\right\},\label{eq:Ssitaware_decompose}\\
                   &=\{\SOne\times
                   2^{\SensorSet\setminus\SOne}|\SOne\in\Ssitawarered\}\label{eq:Ssitaware_decompose_algo}.
    \end{align}
\end{thm}
\begin{proof}
    We show that $\Gamma(\mcS\cup\Stask)-\Gamma(\mcS) =
    \Gamma(\SOne\cup\Stask) - \Gamma(\SOne)$, which implies
    that $\Gamma(\mcS\cup\Stask)-\Gamma(\mcS)=0$ if and only
    if $\Gamma(\SOne\cup\Stask)-\Gamma(\SOne)=0$. This
    implies \eqref{eq:Ssitaware_decompose} by
    Proposition~\ref{prop:Ssitaware_defn}, and
    \eqref{eq:Ssitaware_decompose_algo} follows from
    \eqref{eq:Ssitaware_decompose}. The complete proof is in
    Appendix~\ref{app:thm_ssitware}.
\end{proof}
\newcommand{\proofThmSsitaware}{
    We have to show that
    \begin{align}
        \Ssitaware &= \left\{\mcS\in\SensorSet\middle|
            \begin{array}{c}
                \SOne=\mcS\cap\Sred,\\
                \Gamma(\SOne\cup\Stask)=\Gamma(\SOne)
            \end{array}\right\},\label{eq:Ssitaware_decompose_rep}.
    \end{align}

    We show that $\Gamma(\mcS\cup\Stask)-\Gamma(\mcS) =
    \Gamma(\SOne\cup\Stask) - \Gamma(\SOne)$, which implies
    that $\Gamma(\mcS\cup\Stask)-\Gamma(\mcS)=0$ if and only
    if $\Gamma(\SOne\cup\Stask)-\Gamma(\SOne)=0$. This
    implies \eqref{eq:Ssitaware_decompose_rep} by
    Proposition~\ref{prop:Ssitaware_defn}.
    
    For any $\mcS\in\PowerSet$, we use \eqref{eq:Sred} and
    \eqref{eq:Ssitaware_decompose_rep} to define
    \begin{align}
        \SOne=\mcS\cap\Sred\in 2^\Sred\mbox{ and
        }\STwo=\mcS\setminus\SOne.
    \end{align}

    \emph{Proof for $\Gamma(\mcS\cup\Stask)-\Gamma(\mcS) =
    \Gamma(\SOne\cup\Stask) - \Gamma(\SOne)$}: We will
    use the claims:\\
    1) ${\dimR{\colTSTwo\cap \colTStask}=0}$, and\\
    2) $\dimR{\colTSOne\cap \colTSTwo\cap\colTStask}=0$.\\
    On applying Lemma~\ref{lem:prop_TS_gamma}\ref{lem:prop_TS_gamma_cup_row} twice,
    \begin{align}
        \Gamma(\mcS\cup\Stask)
        &=\Gamma(\SOne\cup\STwo\cup\Stask) \nonumber \\
        &=\Gamma(\SOne) + \Gamma(\STwo) + \Gamma(\Stask)
        \nonumber \\
        &\ - \dimR{\colTSOne\cap \colTSTwo}  \nonumber \\
        &\ - \dimR{\colTSOne\cap \colTStask} \nonumber \\
        &\ - \dimR{\colTSTwo\cap \colTStask} \nonumber \\
        &\ + \dimR{\colTSOne\cap \colTSTwo\cap\colTStask}
        \label{eq:Gamma_S_Stask}.
    \end{align}
    By our claims above, the last two terms in
    \eqref{eq:Gamma_S_Stask} is zero.  By adding and
    subtracting an additional $\Gamma(\SOne)$ to
    \eqref{eq:Gamma_S_Stask}, we have
    $\Gamma(\mcS\cup\Stask)=\Gamma(\SOne\cup\STwo) +
    \Gamma(\SOne\cup\Stask) - \Gamma(\SOne)$.  Since
    $\mcS=\SOne\cup\STwo$, we have
    $\Gamma(\mcS\cup\Stask)-\Gamma(\mcS) =
    \Gamma(\SOne\cup\Stask)-\Gamma(\SOne)$.

    \emph{Proof of claim 1)}: By \eqref{eq:Sred_dim}, for
    any $s\in\SensorSet\setminus\Sred$,
    $\dimR{\colTsmalls\cap \colTStask}=0$. Therefore, for
    any $s_1,s_2\in \SensorSet\setminus \Sred$,
    $\dimR{\colTsmallsOne\cap
    \colTStask}=\dimR{\colTsmallsTwo\cap
\colTStask}=\dimR{\colTsmallsOne\cap\colTsmallsTwo\cap
\colTStask}=0$. By applying
Lemma~\ref{lem:prop_TS_gamma}\ref{lem:prop_TS_gamma_cup_row}
twice, we have $\Gamma(\{s_1\} \cup \{s_2\} \cup
\Stask)=\Gamma(\{s_1\} \cup \{s_2\}) + \Gamma(\Stask)$.
    This also implies that $\dimR{\colTsOneOrsTwo\cap
    \colTStask}=0$. Using similar arguments inductively, we
    conclude that ${\dimR{\colTSTwo\cap \colTStask}=0}$
    since $\STwo\subseteq2^{\SensorSet\setminus \Sred}$.

    \emph{Proof of claim 2)}: Clearly, $\colTSOne\cap
    \colTSTwo\cap\colTStask$ is a subset of $
    \colTSTwo\cap\colTStask$, which implies that
    $\dimR{\colTSOne\cap\colTSTwo\cap \colTStask}$ is
    smaller than $\dimR{\colTSTwo\cap \colTStask}$, by
    definition. Since  ${\dimR{\colTSTwo\cap \colTStask}=0}$
    (shown in claim 1), we have $\dimR{\colTSOne\cap\colTSTwo\cap
    \colTStask}=0$.
}

Table~\ref{tab:table_example} shows $\Ssitawarered$ for the
illustrative example, with %
$| \Ssitawarered| =
6$ and  $|\Ssitaware| = 12$.  
For this problem, $2^\Sred$ has only 7 elements, while $\PowerSet$ has 15.
The computational savings become far more dramatic
for larger problems, as illustrated %
in Section~\ref{sec:app}.

\begin{lem}\label{lem:Stask_feas}
    $\Stask$ is a subset of $\Sred$, and $\Stask$ is a
    member of $\Ssitawarered$ and $\Ssitaware$.
\end{lem}
    \newcommand{\proofLemStaskFeas}{
    By \eqref{eq:Sred},
    $\Gamma(s\cup\Stask)=\Gamma(\Stask)<\Gamma(s)+\Gamma(\Stask)$
    for each $s\in\Stask$. Thus, $\Stask\in\Sred$.

    By construction,
    $\Gamma(\Stask\cup\Stask)=\Gamma(\Stask)$. Thus,
    $\Stask\in\Ssitawarered$.
   } 
Lemma~\ref{lem:Stask_feas} describes the intuitive
observation that constructing a user-interface using only
the sensors that describe the task should also be sufficient
to complete the task. The proof of Lemma~\ref{lem:Stask_feas} is given in Appendix~\ref{app:lem_Stask_feas}.

\subsection{User trust in the automation}

\begin{figure}[t]
    \centering
    \includegraphics[width=0.5\linewidth]{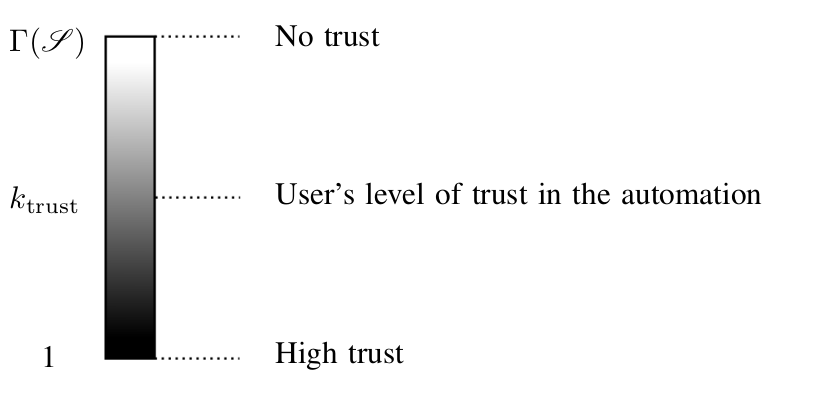} 
    \caption{ The parameter $\ktrust$ indicates the level of
    trust the user has in the automation, with high values
corresponding to low trust, and vice versa.  }
    \label{fig:ktrust}
\end{figure}

User trust in the
automation depends on many factors, including the
expertise of the user, the performance and reliability of the automation,
and the difficulty of the task.  While some dimensions of trust may be static (i.e., dispositional trust), other dimensions may be highly dynamic (i.e., situational or learned trust) \cite{Hoff15,Jain2019}. 
Both low and high levels of trust in the automation are known to be problematic, as they are related to disuse of the automation due to under-reliance, and misuse due to over-reliance, respectively \cite{ParasuramanRiley97}. 

The main principle driving the trust constraint \eqref{eq:prob_defn_constraint2} is that 
the information presented to the user should to be responsive to, and appropriate for, the user's current level of trust in the automation \cite{LeeSee04}.  We focus on the challenges associated with low levels of trust, although extensions to overtrust may be possible.  We presume that additional information would be helpful when the user's trust in the automation is relatively low, but that when the user's trust is relatively high, additional information is not warranted, and may actually be detrimental, if it is overwhelming to the user \cite{Billings97,PSW08}.  
Mathematically, we account for this phenomenon by constraining the user information index by the user's level of trust in the automation.  

\begin{defn}[\textbf{Trust constraint}]\label{defn:Strust}
For a given level of trust in the automation, described by $\ktrust\in \mathbb{N}_{[1,\Gamma(\SensorSet)]}$, we define the set of sensors that are compatible with trust level $\ktrust$ as those whose user information index is above $\ktrust$:
\begin{align}
    \Strust&=\{\mcS\in\PowerSet: \Gamma(\mcS)\geq \ktrust\}\label{eq:Strust}.
\end{align}
\end{defn}

The trust level $\ktrust$ could correspond to a variety of trust metrics, depending on the problem at hand \cite{LewisFTA2018}.  Although considerable variability exists amongst questionnaire-based trust metrics \cite{SATI, HTC, ED}, many seek a summative assessment of trust.  For example, in the SHAPE Automation Trust Index instrument, the `overall amount of trust in the total' system, which solicits trust as a percentage, would be most relevant to our framework \cite{SATI}.  A quantized, affine transformation from the SATI scale, ranging from $0\%$ (no trust) to 100\% (full trust), to our trust level scale, ranging from $\Gamma(\SensorSet)$ (low trust) to 1 (high trust), respectively, would map the SATI `overall trust' to our trust level $\ktrust$, resulting in a static value for a given user.  A similar transformation could be applied to several recent efforts in dynamic trust sensing via behavioral \cite{Sadrfaridpour2017,setter2016trust} and psychophysiological data \cite{Jain2019} (which feature either real-valued, bounded trust variables, i.e., $T(t) \in [0, 1]$ for some trust value $T(t)$, or discrete-valued trust variables, i.e., `low', `medium', `high'), which would allow $\ktrust$ to vary over time.

Because our system model (Figure \ref{fig:manualcontrol}) presumes that the user dictates high level reference tracking, and the automation carries out low-level control, Assumption \ref{assum:AutoWellDesigned} in effect implies that the 
the user
delegates the control of the unobservable state $\eta(t)\in
\mathbb{R}^{n-\Gamma(\mcS)}$ to the automation.  Hence by imposing
a lower bound on the user information index $\Gamma(\cdot)$
in \eqref{eq:Strust}, we impose an upper bound on the
dimension of the unobservable states.  In essence, this bound 
ensures that the unobservable state space doesn't become so large 
that it causes further decrease in trust. 

For example, in off-nominal operation (i.e., scenarios in which the user
may not trust the automation), high values of $\ktrust$
ensure that unobservable state is low dimensional, and the
user retains a large degree of control. On the other hand,
in nominal operation, low values of
$\ktrust$ allow the dimension of the unobservable states to increase, potentially reducing cognitive workload as 
the user delegates control over these variables to the automation.

Table~\ref{tab:table_example} shows $\Strust$ for the
illustrative example %
under various levels of trust. 
We see that
$\Strust=\PowerSet$ when $\ktrust=1$, meaning that all possible interfaces satisfy the trust constraint when the user's trust level is high.  With higher $\ktrust$ (i.e., lower trust level), 
number of sensor combinations that need to be considered
for the user-interface design drastically reduces. For $\ktrust=4$, only four user-interfaces are feasible; 
the observable state is zero-dimensional for these user-interfaces.

\subsection{Dynamics-driven user-interface design as
tractable, combinatorial optimization problems}
\label{sub:UID}

With the situation awareness constraint
\eqref{eq:Ssitaware_third} and trust constraint
\eqref{eq:Strust} established, we reformulate
\eqref{prob:prob_defn} as the combinatorial optimization
problem,
\begin{subequations}
    \begin{align}
        \underset{\mcS\in\PowerSet}{\mathrm{minimize}}&\quad |\mcS| \label{eq:UI_orig_cost}\\
        \mathrm{subject\ to}&\quad
        \Gamma(\mcS)=\Gamma(\mcS\cup\Stask)\label{eq:UI_orig_sitaware}\\
                            &\quad
        \Gamma(\mcS)\geq\ktrust\label{eq:UI_orig_trust}
    \end{align}\label{prob:UI_orig}%
\end{subequations}%
This problem
is well-posed, since $\SensorSet$ is a feasible solution:
$\Gamma(\SensorSet)=\Gamma(\SensorSet\cup\Stask)$, and
$\Gamma(\SensorSet)\geq \ktrust$, by definition. In other
words, the user-interface constructed using all the sensors in $\SensorSet$ is always a feasible solution to \eqref{prob:UI_orig}, irrespective of the task $\Stask$ and the value of $\ktrust$. 

However, solving
\eqref{prob:UI_orig} directly is hard, owing to the
potentially large number of sensor combinations in
consideration $\PowerSet$. We propose different tractable
methods to solve \eqref{prob:UI_orig} using the properties
of $\Gamma(\cdot)$.
First, using Theorem~\ref{thm:Ssitaware_decompose}, we reformulate \eqref{prob:UI_orig}
into \eqref{prob:UI} without introducing any approximation,
\begin{subequations}
    \begin{align}
        \underset{\mcS\in\PowerSet,\SOne\in2^\Sred}{\mathrm{minimize}}&\quad |\mcS| \label{eq:UI_cost}\\
        \mathrm{subject\ to}\hspace*{1.1em}&\quad \SOne\in\Ssitawarered\label{eq:UI_PQ}\\
                            &\quad \SOne=\mcS\cap\Sred\label{eq:UI_sitaware}\\
                            &\quad \Gamma(\mcS)\geq\ktrust\label{eq:UI_trust}
    \end{align}\label{prob:UI}%
\end{subequations}%
We denote the optimal solution of \eqref{prob:UI} as
$\mcS^\ast$ and $\SOne^\ast$.

Next, we investigate submodularity and monotonicity of
$\Gamma(\cdot)$, 
since 
these properties enable greedy heuristics for computing efficient,
near-optimal solutions (see Appendix \ref{app:sub_max}). %
We refer the reader
to~\cite{summers_2014,
clark_submodularity_2017,fujishige2005submodular,Lovasz1983,krause2014submodular}
for more details.

\begin{defn}[\textbf{Submodularity}]
A set function $f(\cdot)$ is submodular if for all sets
$\SOne, \STwo\in \PowerSet$, 
\begin{align}
    f(\SOne) + f(\STwo) &\geq f(\SOne \cup \STwo) + f(\SOne
    \cap \STwo).\label{eq:subm_defn}
\end{align}
\end{defn}
\begin{defn}[\textbf{Monotone increasing}]
A set function $f(\cdot)$ is monotone increasing if for all
sets $\SOne, \STwo\in \PowerSet$, 
\begin{align}
    \SOne \subseteq \STwo & \Rightarrow f(\SOne) \leq f(\STwo).\label{eq:monotone_defn}
\end{align}
\end{defn}

Submodular functions demonstrate diminishing returns, i.e.,
adding an element to a smaller set results in a higher gain
as compared to adding it to a larger set. Monotone increasing functions
preserve the inclusion ordering in $\PowerSet$.

\begin{prop} %
    The user information index $\Gamma(\cdot)$ is a
    submodular monotone increasing
    function.\label{prop:Gamma_subm}
\end{prop}
\newcommand{\proofPropGammaSubm}{
    \emph{Submodularity}:
    For any $\SOne,\STwo\in\PowerSet$, we show that $\Gamma(\cdot)$ meets
    \eqref{eq:subm_defn} using
    Lemma~\ref{lem:prop_TS_gamma}\ref{lem:prop_TS_gamma_cup_row} and
    Lemma~\ref{lem:prop_TS_gamma}\ref{lem:prop_TS_gamma_cap},
    \begin{align}
        \Gamma(\SOne\cup \STwo)&=\Gamma(\SOne) + \Gamma(\STwo) - \dimR{\colTSOne\cap \colTSTwo} \leq\Gamma(\SOne) + \Gamma(\STwo) - \Gamma(\SOne\cap \STwo).\label{eq:subm_proof_2}
    \end{align}
    \emph{Monotone increasing property}: Follows from
    Lemma~\ref{lem:prop_TS_gamma}\ref{lem:prop_TS_gamma_row_TS_monotone}.
}
\begin{proof}
\proofPropGammaSubm{}
\end{proof}

\begin{corr}
    For any $\mcS\in\Ssitaware$,
    $\Gamma(\mcS)\geq\Gamma(\Stask)$.
    \label{corr:lower_bound_mcS}
\end{corr}
\newcommand{\proofCorrLowerBoundMcS}{
    Since $\Stask\subseteq (\mcS\cup\Stask)$, we have $
    \Gamma(\mcS)=\Gamma(\mcS\cup\Stask)\geq\Gamma(\Stask)$
    for every $\mcS\in\Ssitaware$
    \eqref{eq:Ssitaware_third}, due to the monotone
    increasing property of $\Gamma(\cdot)$.
}
\begin{corr}%
    \label{corr:Strust}
    Given $\Stask\in\PowerSet$,
    {\renewcommand{\theenumi}{\alph{enumi}}
    \begin{enumerate}
        \item if $\ktrust\leq \Gamma(\Stask)$, then
            $\Ssitaware\subseteq\Strust$.\label{corr:Strust_ssitawareonly}
        \item if $\ktrust= \Gamma(\SensorSet)$, then
            $\Strust\subseteq\Ssitaware$.\label{corr:Strust_strustonly}
\end{enumerate}}
\end{corr}
\newcommand{\proofCorrStrust}{
    To prove~\ref{corr:Strust_ssitawareonly}), we note that
    for any $\mcS\in\Ssitaware$ with
    $\ktrust\leq\Gamma(\Stask)$, we have
    $\ktrust\leq\Gamma(\Stask)\leq\Gamma(\mcS)$ by
    Corollary~\ref{corr:lower_bound_mcS}. Thus,
    $\mcS\in\Strust$.

    To prove~\ref{corr:Strust_strustonly}), we note that for
    any $\mcS\in\Strust$, $\ktrust\leq\Gamma(\mcS)\leq
    \Gamma(\SensorSet)$. Since $\ktrust=
    \Gamma(\SensorSet)$,  $\Gamma(\mcS)=\Gamma(\SensorSet)$
    for any $\mcS\in\Strust$.  By the monotone increasing
    property of $\Gamma(\cdot)$
    (Proposition~\ref{prop:Gamma_subm}), we have
    $\Gamma(\SensorSet) = \Gamma(\mcS) \leq
    \Gamma(\mcS\cup\Stask)\leq\Gamma(\SensorSet)$ which
    implies $\Gamma(\mcS)=\Gamma(\mcS\cup\Stask)$. Thus,
    $\mcS\in\Strust$ with $\ktrust=\Gamma(\SensorSet)$
    implies $\mcS\in\Ssitaware$ by
    \eqref{eq:Ssitaware_third}.
}

Corollary~\ref{corr:lower_bound_mcS} provides a lower bound
on $\Gamma(\mcS)$ for $\mcS\in\Ssitaware$.
Corollary~\ref{corr:Strust}\ref{corr:Strust_ssitawareonly}
states that the trust constraint \eqref{eq:Strust} is
trivially satisfied, if the user interface enables
situation awareness and $\ktrust$ is low enough (i.e., user's trust level is high enough). On other hand, when $\ktrust$ is as high as possible
(i.e., lowest trust level), user-interface design is
\emph{task-agnostic}, and trust constraint satisfaction
automatically enables situation awareness. 

For the illustrative example given in
Section~\ref{sec:prem}, note that
$\Gamma(\Stask)=\Gamma(\{s_v\})=2$ in
Table~\ref{tab:table_example}.  As stated in
Corollary~\ref{corr:Strust}, we see that $\Ssitaware\subset
\Strust$ when $\ktrust\leq  2 =\Gamma(\Stask)$. Further,
$\Strust\subset \Ssitaware$, when $\ktrust = 4 =
\Gamma(\SensorSet)$.

We propose three different approaches to compute a solution
to \eqref{prob:UI} under different ranges of $\ktrust$.

\subsubsection{An optimal solution 
when $\ktrust\leq \Gamma(\Stask)$}

With a high level of trust, by
Corollary~\ref{corr:Strust}\ref{corr:Strust_ssitawareonly},
the trust constraint \eqref{eq:UI_trust} is trivially
satisfied by any choice of $\mcS\in\Ssitaware$.  Because we seek
sensor combinations with minimum
cardinality, we search only 
in $\Ssitawarered$. Since $\Stask$ is a feasible
solution by Lemma~\ref{lem:Stask_feas}, we can reformulate
\eqref{prob:UI} into \eqref{prob:nom} without introducing
any approximation: %
\begin{subequations}
    \begin{align}
        \underset{\SOne\in\Ssitawarered}{\mathrm{minimize}}
        &\quad | \SOne| \label{eq:nom_cost}\\
        \mathrm{subject\ to}
        &\quad |\SOne|\leq|\Stask|
        \label{eq:nom_card}
    \end{align}\label{prob:nom}%
\end{subequations}%
The optimal solution to (\ref{prob:nom}) is also the optimal solution to (\ref{prob:UI}).  
The constraint \eqref{eq:nom_card} requires a brute force search, hence the
numerical implementation in Algorithm \ref{algo:CP_opt}
has a worst-case computation complexity of $
\mathcal{O}\left(\sum_{i=1}^{| \Stask|}
{{|\Sred|}\choose{i}}\right)$.

\begin{algorithm}[t]
    \caption{Optimal solution to \eqref{prob:UI} when $\ktrust\leq
    \Gamma(\Stask)$}\label{algo:CP_opt}
    \begin{algorithmic}[1]
        \Require~Set of all sensors $\SensorSet$, trust parameter $\ktrust$,
        sensors that define the task $\Stask$, the user information index
        $\Gamma(\cdot)$ 
        \Ensure~An optimal solution to \eqref{prob:UI}
        \State Compute $\Ssitawarered$ using Algorithm~\ref{algo:Enum}         \label{line:algo_CP_opt_saware_a1}
        \State
        $\SensorSet_\mathrm{feas,high-trust}\triangleq
        \Ssitawarered\cap\{\mcS: |\mcS| \leq
        |\Stask|\}$\label{line:algo_CP_opt_saware}
        \State $\mcS^\ast\gets\min\{|\mcS|:
        \mcS\in\SensorSet_\mathrm{feas,high-trust}\}$
        \State\Return $\mcS^\ast$
  \end{algorithmic}
\end{algorithm}

\subsubsection{A greedy suboptimal solution 
when $\ktrust= \Gamma(\SensorSet)$}

With the lowest level of trust, the situation
awareness constraints \eqref{eq:UI_PQ} and
\eqref{eq:UI_sitaware} are trivially satisfied for any
$\mcS\in\Strust$ by
Corollary~\ref{corr:Strust}\ref{corr:Strust_strustonly}. By Proposition~\ref{prop:Gamma_subm}, \eqref{prob:UI} simplifies to the following submodular optimization problem (Appendix~\ref{app:sub_max}),
\begin{subequations}
    \begin{align}
        \underset{\mcS\in\PowerSet}{\mathrm{minimize}}&\quad
        | \mcS| \label{eq:UI_subm_cost}\\
        \mathrm{subject\ to}&\quad
        \Gamma(\mcS)\geq\Gamma(\SensorSet)\label{eq:UI_subm_constraint1}.
    \end{align}\label{prob:UI_subm}%
\end{subequations}%
We compute a suboptimal solution with
provable suboptimality guarantees via a greedy algorithm
(Algorithm~\ref{algo:greedy} in Appendix~\ref{app:sub_max}).

\subsubsection{A suboptimal solution 
for $\Gamma(\Stask)<\ktrust<\Gamma(\SensorSet)$}

For trust values in between, 
we propose Algorithm~\ref{algo:CP}, which 
solves a submodular optimization problem (\eqref{prob:subm_max}, Appendix \ref{app:sub_max}) for every
$\SOne\in\Ssitawarered$,
\begin{subequations}
    \begin{align}
        \underset{\STwo_{\SOne}\subseteq \SensorSet\setminus\SOne}{\mathrm{minimize}}&\quad | \STwo_{\SOne}| \label{eq:UI_Q_cost}\\
        \mathrm{subject\ to}&\quad \Gamma(\SOne\cup\STwo_{\SOne})\geq \ktrust\label{eq:UI_Q_constraint1}
    \end{align}\label{prob:UI_Q}%
\end{subequations}%
By Theorem~\ref{thm:Ssitaware_decompose}, the
optimal solution to \eqref{prob:UI} is the minimum
cardinality set in the following: 
\begin{align}
    \SsetSubopt&=\bigcup\nolimits_{\SOne\in\Ssitawarered}
    \{\SOne\cup\STwo_{\SOne}|\STwo_{\SOne}\mbox{ solves
\eqref{prob:UI_Q}}\}\label{eq:SsetSubopt}.
\end{align}
However, solving \eqref{prob:UI_Q} for each
$\SOne\in\Ssitawarered$ is computationally expensive for
large $| \SensorSet|$.  We know that
$\Gamma(\SOne\cup\STwo_{\SOne})$ is a submodular monotone
function in $\STwo_{\SOne}$ for any
$\SOne\in\PowerSet$~\cite[Sec.  1.2]{krause2014submodular}.
Therefore, \eqref{prob:UI_Q} is also a submodular optimization
problem.  
We again use the
greedy approach (Algorithm~\ref{algo:greedy} in
Appendix~\ref{app:sub_max}) to compute a suboptimal solution
$\STwo_\SOne^\dagger$.  Note that lines $3$--$6$ of
Algorithm~\ref{algo:CP} is trivially parallelizable.

\begin{algorithm}[t]
    \caption{A suboptimal solution to \eqref{prob:UI} for $\Gamma(\Stask)<\ktrust<\Gamma(\SensorSet)$}\label{algo:CP}
    \begin{algorithmic}[1]
        \Require~Set of all sensors $\SensorSet$, trust parameter $\ktrust$,
        sensors that define the task $\Stask$, the user information index
        $\Gamma(\cdot)$ 
        \Ensure~A suboptimal solution to \eqref{prob:UI}
        $\SSubopt$
        \State $\SsetSubopt\gets \emptyset$
        \State Compute $\Ssitawarered$ using Algorithm~\ref{algo:Enum}\label{line:algo_CP_saware}
            \For{$\SOne\in\Ssitawarered$}\label{line:greedyforbegin}
                \State Compute $\STwo_{\SOne}^\dagger$ by
                solving \eqref{prob:UI_Q} given $\SOne$
                suboptimally using
                Algorithm~\ref{algo:greedy}
                (see
                Appendix~\ref{app:sub_max})\label{line:Alg3_Alg4}
                \State Add $\SOne\cup\STwo_{\SOne}^\dagger$
                to $\SsetSubopt$
            \EndFor\label{line:greedyforend}
            \State $\SSubopt\gets\min\{|\mcS|:
        \mcS\in\SsetSubopt\}$\label{line:min}        
        \State\Return $\SSubopt$
  \end{algorithmic}
\end{algorithm}

To quantify the suboptimality bound for Algorithm~\ref{algo:CP}, 
we define a real-valued function $\Delta_\Gamma:
\mathbb{N}_{[1, \Gamma(\SensorSet)]}\times
\PowerSet \to \mathbb{R}$ as
\begin{align}
    \Delta_\Gamma(k, \mcS)&= \begin{cases}
        \begin{array}{ll}
            \log\left(\frac{\Gamma(\SensorSet)}{k-\Gamma(\mcS)}\right)
            & \Gamma(\mcS) < k\\ 
            \infty & \mbox{otherwise}.
        \end{array}
    \end{cases}\label{eq:Delta_defn}
\end{align}

\begin{prop}[\textbf{Suboptimality bound for
    Alg.~\ref{algo:CP}}]\label{prop:algo2_bound}
    For $\Gamma(\Stask)<\ktrust<\Gamma(\SensorSet)$,
    Algorithm~\ref{algo:CP} computes a suboptimal solution
    $\SSubopt$ to \eqref{prob:UI} that satisfies
    \begin{align}
        1\leq \frac{|\SSubopt|}{| \Sopt|}&\leq \left(1+ \max_{\SOne\in \Ssitawarered}\Delta_\Gamma(\ktrust,\SOne\cup\STwo_{\SOne}^-)\right)\label{eq:subopt_bound_algoCP}
    \end{align}
    where $\SOne\cup\STwo_{\SOne}^{-}$ is the solution prior
    to the termination step of
    Algorithm~\ref{algo:greedy} in Line~\ref{line:Alg3_Alg4}
    of Algorithm~\ref{algo:CP}.
\end{prop}
\begin{proof}
    Let $\mcS^\ast=\SOne^\ast\cup\STwo_{\SOne^\ast}^\ast$ be
    the (unknown) optimal solution to \eqref{prob:UI}, where
    $\SOne^\ast\subseteq\Ssitawarered$. Such a decomposition
    is guaranteed by Theorem~\ref{thm:Ssitaware_decompose}.
    Let $\STwo_{\SOne^\ast}^{\dagger}$ be the solution of
    \eqref{prob:UI_Q} for $\SOne^\ast$ computed using
    Algorithm~\ref{algo:greedy}, and
    $\STwo_{\SOne^\ast}^{-}$ be the solution prior to the
    termination step. By Lemma~\ref{lem:greedy} in
    Appendix~\ref{app:sub_max},
    \begin{align}
        |\STwo_{\SOne^\ast}^{\dagger}|&\leq|
        \STwo_{\SOne^\ast}^\ast|\left(1+\Delta_{\Gamma}(\ktrust,\SOne^\ast\cup\STwo_{\SOne^\ast}^{-})\right)\label{eq:Stwo_subopt_defn}.
    \end{align}
    Equation \eqref{eq:Stwo_subopt_defn} uses the
    observation that $\Gamma(\emptyset)=0$ and upper bounds the suboptimality bound in Lemma~\ref{lem:greedy} in Appendix~\ref{app:sub_max} using $\Delta_\Gamma$. The upper bound and the finiteness of  $\Delta_\Gamma$ follows from the fact that $\Gamma(\SOne^\ast\cup\STwo_{\SOne^\ast}^\dagger)\geq \ktrust > \Gamma(\SOne^\ast\cup\STwo_{\SOne^\ast}^{-})$ by the termination rule of Algorithm~\ref{algo:greedy}.
    
    By line~\ref{line:min} of Algorithm~\ref{algo:CP}, we have
    \begin{align}
        | \mcS^\ast|=| \SOne^\ast| + | \STwo_{\SOne^\ast}| \leq |\SSubopt|&\leq | \SOne^\ast| + | \STwo_{\SOne^\ast}^{\dagger}|.\label{eq:intermediate_subopt_bound}
    \end{align}
    Applying \eqref{eq:Stwo_subopt_defn} to \label{eq:intermediate_subopt_bounds} and rearranging the
    resulting terms,
    \begin{align}
        1 \leq \frac{|\SSubopt|}{| \mcS^\ast|} \leq \left(1
        + \frac{| \STwo_{\SOne^\ast}^\ast|}{|
    \mcS^\ast|}\Delta_{\Gamma}(\ktrust,\SOne^\ast\cup\STwo_{\SOne^\ast}^{-})\right)
    &\leq \left(1 +
        \Delta_{\Gamma}(\ktrust,\SOne^\ast\cup\STwo_{\SOne^\ast}^{-})\right)
        \nonumber \\
        &\leq \left(1+ \max_{\SOne\in
        \Ssitawarered}\Delta_\Gamma(\ktrust,\SOne\cup\STwo_{\SOne}^-)\right)
        \nonumber
    \end{align}
    since $|\STwo_{\SOne^\ast}^\ast|\leq| \mcS^\ast|$ and
    $\Delta_\Gamma(\ktrust,\SOne^\ast\cup\STwo_{\SOne^\ast}^-)$ is bounded from above by $\displaystyle \max_{\SOne\in
    \Ssitawarered}\Delta_\Gamma(\ktrust,\SOne\cup\STwo_{\SOne}^-)$.
    For every
    $\SOne\in\Ssitawarered$, 
    $\Delta_\Gamma(\ktrust,\SOne\cup\STwo_{\SOne}^-)$ is finite since
    $\Gamma(\SOne\cup\STwo_{\SOne}^-) < \ktrust$ by the termination rule of Algorithm~\ref{algo:greedy}.
\end{proof}

We simplify \eqref{eq:subopt_bound_algoCP} to obtain
a weaker upper bound,
\begin{align}
    | \Sopt|\leq |\SSubopt| \leq | \Sopt| (1+
    \log(\Gamma(\SensorSet))).\label{eq:subopt_bound_algoCP_indep}
\end{align}
Equation \eqref{eq:subopt_bound_algoCP_indep} follows from
the observation that
$\Delta_\Gamma(\ktrust,\SOne\cup\STwo_{\SOne}^-)\leq \log(\Gamma(\SensorSet))$ for every
$\SOne\in\Ssitawarered$. Therefore, $\displaystyle\max_{\SOne\in
\Ssitawarered}\Delta_\Gamma(\ktrust,\SOne\cup\STwo_{\SOne}^-)
\leq \log(\Gamma(\SensorSet)))$.
Equation \eqref{eq:subopt_bound_algoCP_indep} shows that the upper bound in \eqref{eq:subopt_bound_algoCP} can not be arbitrarily loose.

\begin{prop}[\textbf{Computational complexity bound for
    Alg.~\ref{algo:CP})}]\label{prop:algo2_runtime}
    For $\Gamma(\Stask)<\ktrust<\Gamma(\SensorSet)$, Algorithm~\ref{algo:CP} has
    a worst-case computational complexity of $
    \mathcal{O}\left(2^{|\Sred|} {|\SensorSet|}^2\right)$.
\end{prop}
\begin{proof}
    In Algorithm~\ref{algo:CP}, the evaluation of
    lines~\ref{line:algo_CP_saware},~\ref{line:greedyforbegin}--\ref{line:greedyforend},
    and~\ref{line:min} have a worst-case computational
    complexity of $\mathcal{O}(2^{|\Sred|})$,
    $\mathcal{O}(| \Ssitawarered|
    {|\SensorSet|}^2)$ (from Lemma~\ref{lem:greedy}
    in Appendix~\ref{app:sub_max}), and $\mathcal{O}(|
    \Ssitawarered|)$, respectively.  The worst-case
    computational complexity of Algorithm~\ref{algo:CP} is
    $ \mathcal{O}(| \Ssitawarered|
    {|\SensorSet|}^2 + | \Ssitawarered| +
    | 2^\Sred|)$. Using the observation that $|
    \Ssitawarered|\leq2^{|\Sred|}$, we obtain
    the simplified worst-case complexity
    bound.
\end{proof}

An alternative heuristic to Algorithm~\ref{algo:CP} is to
use Algorithm~\ref{algo:CP_opt} to solve \eqref{prob:nom} to
obtain $\SOne^\dagger\subseteq\Ssitawarered$ that enables
situation awareness, and then solve the associated
submodular maximization problem \eqref{prob:UI_Q} with
$\SOne^\dagger$. This approach may provide a faster
solution, since the search for $\SOne^\dagger$ is
assisted by the cardinality constraint
\eqref{eq:nom_card}.  Further, it only requires the solution
of a single submodular maximization problem, as opposed to a
collection of $| \Ssitawarered|$ problems in
Algorithm~\ref{algo:CP}. However, the
suboptimality bound of Algorithm~\ref{algo:CP} no longer
holds for this approach, since
\eqref{eq:intermediate_subopt_bound} fails to hold.  

The approaches proposed in this section are 
summarized in Table~\ref{tab:all_approaches}. 

\begin{table}
    \centering
    \setlength\arraycolsep{5pt}
    \begin{tabular}{|m{4cm}|m{3cm}|m{3cm}|m{3.5cm}|}
    \hline									%
    Range of $\ktrust$ & Method &  Optimality & Worst-case compute\newline complexity \\
    \hline\hline
    $\ktrust$ $\leq\Gamma(\Stask)$  &
    Alg.~\ref{algo:CP_opt}  &  Optimal    & $\mathcal{O}\left(\sum_{i=1}^{| \Stask|} {{|\Sred|}\choose{i}}\right)$ \\ \hline
    $\Gamma(\Stask)$ $<\ktrust$ $<\Gamma(\SensorSet)$ & Alg.~\ref{algo:CP}  & 
    Suboptimal as in \eqref{eq:subopt_bound_algoCP} &
    $\mathcal{O}\left(2^{|\Sred|} {|\SensorSet|}^2\right)$ \\ \hline
    $\ktrust$ $=\Gamma(\SensorSet)$ &  Solve \eqref{prob:UI_subm} via Alg.~\ref{algo:greedy}  & Suboptimal as in \eqref{eq:subopt_bound} &$\mathcal{O}\left({|\SensorSet|}^2\right)$ \\ \hline
   \end{tabular}
   \caption{Solution methods to \eqref{prob:UI}
   for $\ktrust\in[1,\Gamma(\SensorSet)]$.}
   \label{tab:all_approaches}
\end{table}

\section{Efficient implementation of Algorithm~\ref{algo:Enum}}
\label{sec:CP}

We propose a computationally efficient implementation of
Algorithm~\ref{algo:Enum} using constraint programming and a
novel enumeration framework based on binary number
representation. The proposed approach %
exploits the monotonicity properties of the
user information index function.

\subsection{Enumerating $\Ssitawarered$ via constraint
programming}

Constraint programming exploits transitivity properties in
set functions to reduce the search space \cite{rossi2006handbook}. 
(For example,
for a monotone increasing constraint $f(\mcS)\leq k$ for some 
$k\in \mathbb{N}$, infeasibility of $\SOne\in
\PowerSet$ implies infeasibility of all $\STwo\in \PowerSet$
such that $\SOne\subseteq \STwo$.)
To avoid enumeration in 
Algorithm~\ref{algo:Enum} of the set $\Ssitawarered$
\eqref{eq:soneset} using
Theorem~\ref{thm:Ssitaware_decompose}, 
we construct the feasibility problem corresponding to (16), 
\begin{subequations}
    \begin{align}
        \mathrm{find\ all}         
        & \hspace*{0.5em}
            \begin{array}{l}
                \SOne\in2^{\Sred},\\
                t\in
                \mathbb{N}_{[\Gamma(\Stask),\Gamma(\Sred)]}
            \end{array} \label{eq:soneset_cost}\\
        \mathrm{subject\ to}  
        &\quad \Gamma(\SOne) \geq
        t\label{eq:soneset_GammaSone}\\
        & \quad\Gamma(\SOne\cup\Stask) \leq
        t\label{eq:soneset_GammaSoneCupStask}
    \end{align}\label{prob:feas_soneset}%
\end{subequations}%
Since $\Gamma(\cdot)$ is 
monotone increasing (Proposition~\ref{prop:Gamma_subm}), we can
prune the search space when a tuple $(\SOne,t)$ that
does not satisfy \eqref{eq:soneset_GammaSoneCupStask} is
encountered. Specifically, given $\SOne_1\in\Sred$ such that
$\Gamma(\SOne_1\cup\Stask)\not\leq t_0$ for some $t_0\in
\mathbb{N}_{[\Gamma(\Stask),\Gamma(\Sred)]}$, then for every
superset $\SOne_2\in\Sred$, $\SOne_1\subseteq\SOne_2$ and
$t\leq t_0$, we know $\Gamma(\SOne_2\cup\Stask)\not\leq t$.
We can also incorporate the cardinality constraint
\eqref{eq:nom_card} to further restrict the search space in
Algorithm~\ref{algo:CP_opt}. 

\begin{prop}
    A set $\SOne\subseteq \Sred$ is feasible for
    \eqref{prob:feas_soneset} for some
    $t\in\mathbb{N}_{[\Gamma(\Stask),\Gamma(\Sred)]}$ if and
    only if $\SOne\in\Ssitawarered$.
\end{prop}
\begin{proof}
    The constraints \eqref{eq:soneset_GammaSone} and \eqref{eq:soneset_GammaSoneCupStask} together are equivalent to the following equality constraint (identical to \eqref{eq:soneset}), $$t=\Gamma(\SOne)=\Gamma(\SOne\cup\Stask).$$ 
    We have $t\in
        \mathbb{N}_{[\Gamma(\Stask),\Gamma(\Sred)]}$, 
    since $\Gamma(\cdot)$ is monotone increasing (Proposition~\ref{prop:Gamma_subm}) and  Corollary~\ref{corr:lower_bound_mcS}.
\end{proof}

\subsection{Computationally efficient enumeration of the search space}
\label{sub:gen}

We employ constraint propagation to enumerate the search space, which, for ease of discussion, we presume is $\PowerSet$.
We desire to
create an oracle, referred to as a \emph{generator}, that
provides the next sensor combination in $\PowerSet$ which
needs to be evaluated. 
The generator must
satisfy three requirements:
\begin{enumerate}
    \item[R1)] produce sensor combinations within $\PowerSet$ in an exhaustive manner, 
    \item[R2)] eliminate sensor combinations that are a superset of a given set, and 
    \item[R3)] enforce cardinality constraints.
\end{enumerate}

Any sensor combination $\mcS\in \PowerSet$ can be associated
with a unique $| \SensorSet|$-bit long binary number
representation, with the bit values set to one at the
respective positions of every selected sensor. We also use
a bijection of this representation to the corresponding
decimal number $N_{\mcS}\in \mathbb{N}_{\left[0, 2^{|
\SensorSet|} -1\right]}$. Therefore, any generator over
$\mathbb{N}_{\left[0, 2^{| \SensorSet|} -1\right]}$
exhaustively enumerates $\PowerSet$.  

A naive approach to enumerate $\PowerSet$ is to use a linear generator, which enumerates $\mathbb{N}_{\left[0, 2^{| \SensorSet|} -1\right]}$ by incrementing $N_{\mcS}$ by $1$.
However, satisfying R2) and R3) with a linear generator is difficult.
We propose a generator that satisfies all three requirements by enumerating over a tabular representation of $\PowerSet$.
We associate a unique column number $\coln_{\mcS}\in \mathbb{N}_{[0,| \SensorSet|-1]}$ and row number $\rown_{\mcS}\in \mathbb{N}_{\left[0, 2^{| \SensorSet|-1} -1\right]}$ with every sensor combination $\mcS\in\PowerSet$,
\begin{subequations}
    \begin{align}
        \coln_{\mcS}&=\lfloor\log_2( N_{\mcS})\rfloor\label{eq:coln_defn}\\
        \rown_{\mcS}&= N_{\mcS}-2^{\coln_{\mcS}}\label{eq:rown_defn}
    \end{align}\label{eq:table_defn}%
\end{subequations}%
where $\lfloor a\rfloor$ is the floor of $a\in \mathbb{R}$, the largest integer below $a$.
The column number is the position of the most significant bit of the binary representation of $N_{\mcS}$, and the row number is the decimal representation of the number defined by the remaining bits.
The number of non-zero bits in the binary representation of $\mcS$ is equal to $|\mcS|$.

We demonstrate this approach on 
$\SensorSet=\{s_0,s_1,s_2,s_3,s_4\}$ in 
Table~\ref{tab:bin_tab}.  For illustration, 
consider $\mcS=\{s_0,s_4\}$.  We associate with
$\mcS$ a binary representation, $10001$, based on the
selection of sensors. The decimal number representation of
$10001$ is $N_\mcS=17$.  Note that $| \mcS| = 2$ is the
number of non-zero bits in the binary representation
$10001$. By \eqref{eq:coln_defn}, $\coln_{\mcS}=4$ which is
the position (count starts from zero) of the most
significant non-zero bit.  By \eqref{eq:rown_defn},
$\rown_{\mcS}=1 = 17 - 2^4$.

\begin{table}[!ht]
\centering
\begin{tabular}{l|lllll}
Row & $s_0$ & $s_1$ & $s_2$ & $s_3$ & $s_4$ \\ \hline
 0  & 1     & 2     & 4     & 8     & 16    \\
 1  &       & 3     & 5     & 9     & 17    \\
 2  &       &       & 6     & 10    & 18    \\
 3  &       &       & 7     & 11    & 19    \\
 4  &       &       &       & 12    & 20    \\
 5  &       &       &       & 13    & 21    \\
 6  &       &       &       & 14    & 22    \\
 7  &       &       &       & 15    & 23    \\
 8  &       &       &       &       & 24    \\
 9  &       &       &       &       & 25    \\
10  &       &       &       &       & 26    \\
11  &       &       &       &       & 27    \\
12  &       &       &       &       & 28    \\
13  &       &       &       &       & 29    \\
14  &       &       &       &       & 30    \\
15  &       &       &       &       & 31    
\end{tabular}
\caption{Binary Iteration Table for $\SensorSet=\{s_0,s_1,s_2,s_3,s_4\}$}
\label{tab:bin_tab}
\end{table}

\subsubsection{Satisfaction of R1}

All numbers in $\mathbb{N}_{\left[0, 2^{| \SensorSet|} -1\right]}$ have a unique position in the tabular representation of $\PowerSet$ which follows from the unique binary representation of $N_{\mcS}$ by \eqref{eq:table_defn}.
Thus, any enumeration of the proposed table satisfies R1.

\subsubsection{Satisfaction of R2}

Due to \eqref{eq:rown_defn}, each row contains sensor
combinations with a similar pattern in the lower significant
bits.  Specifically, the binary representation of the row
number coincides with the binary representation of
$N_{\mcS}$ \emph{without} its most significant bit. For example,
row $3$ of Table~\ref{tab:bin_tab} contains numbers $7$
(select $s_0,s_1,s_2$), $11$ (select $s_0,s_1,s_3$), and
$19$ (select $s_0,s_1,s_4$).  All these numbers have the
elements $s_0$ and $s_1$ in common, since their row number,
$3$, has the binary representation $00011$.

Using this observation, we skip enumeration of the
supersets of infeasible sets, by maintaining a collection of
rows to skip. For example, suppose we
wish to skip enumeration of all supersets of
$\SOne=\{s_0,s_1\}$. This is the case when $\SOne$
violates \eqref{eq:soneset_GammaSoneCupStask}. 
We must skip rows $3,7,11$, and $15$ of Table~\ref{tab:bin_tab},
as they are the row numbers with the pattern $XX11$ where
$X$ indicates ``don't care'' bits.  %
The generator then produces $\mcS$, with
$N_{\mcS}\not\in\{7,11,19,15,23,27,31\}=\left\{00111,01011,10011,01111,10111,11011,11111\right\}$.  Note that each of the skipped
numbers have the bits set at their zeroth and first
positions, i. e., they are supersets of $\SOne$.

\subsubsection{Satisfaction of R3}

Recall that the binary representation of $\rown_{\mcS}$
provides an accurate characterization of $\mcS$, except for
one sensor element. Therefore, the number of non-zero bits
in the binary representation of $\rown_{\mcS}$ is equal to
$| \mcS| -1$, since the most significant bit is excluded. 
Thus, by restricting the number of bits in the binary
representation of the enumerated row numbers, we can enforce
cardinality constraints like
\eqref{eq:nom_card} and satisfy R3.

The proposed generator provides an efficient way to
enumerate the search space and incorporate constraint
programming. Specifically, the row-wise enumeration permits
the enforcement of cardinality constraints as well as
the elimination of supersets of an infeasible sensor
combination. We use this framework for computations
involving Algorithm~\ref{algo:Enum}, including computation of
the set $\SsetSubopt$ in
Algorithm~\ref{algo:CP_opt}, and 
the set $\Ssitawarered$ in 
Algorithm~\ref{algo:CP}.

\section{Application: User-interface design for IEEE 118-Bus Power Grid}
\label{sec:app}

\begin{table*}
    \setlength\arraycolsep{1pt}
    \newcommand{\networkWidth}{2.5cm}
    \newcommand{\GsWidth}{1.2cm}
    \newcommand{\userTrust}{4cm}
    \newcommand{\userSolution}{3.5cm}
    \newcommand{\so}{}
    \newcommand{\op}{}
    \centering
    \normalsize
    \adjustbox{width=1\textwidth}{
    \begin{tabular}{|c|c|c|c|l|c|c|}
    \hline				 
    \multirow{2}{*}{\begin{minipage}{\networkWidth}\centering Network configuration \end{minipage}} & \multirow{2}{*}{\begin{minipage}{\userTrust}\centering Trust level ($\ktrust$) \end{minipage}} &  \multicolumn{4}{c|}{User-interface design} & Compute \\\cline{3-6}
                                                                                                    &
                                                                                                    &
    Approach & $| \Ssitawarered|$ & {\footnotesize
    $|\mcS_\mathrm{soln}|\leq\Delta{|\mcS^\ast|}$} & Solution &
    Time \\ \hline\hline

    \multirow{3}{*}{\begin{minipage}{\networkWidth} Normal operation $(A_1,B_1,\CtaskOne)$  $\Gamma(\StaskOne)=34$ \end{minipage}}&  High ($\ktrust = 24$) &  Alg.~\ref{algo:CP_opt} &  $1$ & \op{}   $\Delta=1$ & $\StaskOne$ & $0.23$ s \\\cline{2-7}
                                                                                                                                                      & Moderate ($\ktrust = 44$) &  Alg.~\ref{algo:CP} &  $1$ &  \so{}   $\Delta=4.09$ & $\StaskOne\cup\{1,2,3,5,9\}$ & $0.91$ s \\\cline{2-7}
                                                                                                                                                      & None ($\ktrust = 108$) &  Alg.~\ref{algo:greedy} &  -- & \so{} $\Delta=4.99$ & $\SensorSet$ & $9.62$ s \\\hline\hline

    \multirow{3}{*}{\begin{minipage}{\networkWidth} Bus $38$ down $(A_2,B_1,\CtaskTwo)$  $\Gamma(\StaskTwo)=14$ \end{minipage}}&  High ($\ktrust = 4$) &  Alg.~\ref{algo:CP_opt} &  $1$ & \op{}   $\Delta=1$ & $\StaskOne$ & $0.11$ s \\\cline{2-7}
   & Moderate ($\ktrust = 24$) &  Alg.~\ref{algo:CP} &  $1$ &  \so{}   $\Delta=3.49$ & $\StaskTwo\cup\{1,2,3,4,5\}$ & $0.49$ s \\\cline{2-7}
   & None ($\ktrust = 108$) &  Alg.~\ref{algo:greedy} & -- & \so{}   $\Delta=4.99$ & $\SensorSet$ & $9.47$ s \\\hline\hline

    \multirow{3}{*}{\begin{minipage}{\networkWidth} {Line $65,66$ down} $(A_3,B_1,\CtaskThree)$  $\Gamma(\StaskThree)=30$ \end{minipage}}&  High ($\ktrust = 20$) &  Alg.~\ref{algo:CP_opt} &  $1$ & \op{}   $\Delta=1$ & $\StaskThree$ & $0.20$ s \\\cline{2-7}
   & Moderate ($\ktrust = 40$) &  Alg.~\ref{algo:CP} &  $1$ &  \so{}   $\Delta=4.00$ & $\StaskThree\cup\{1,2,3,5,9\}$ & $0.77$ s \\\cline{2-7}
   & None ($\ktrust = 108$) &  Alg.~\ref{algo:greedy} & -- & \so{}   $\Delta=4.99$  &  $\SensorSet$ & $9.26$ s \\\hline\hline

    \multirow{5}{*}{\begin{minipage}{\networkWidth} Alternate generators down $(A_1,B_2,\CtaskFour)$  $\Gamma(\StaskFour)=52$ \end{minipage}}&  High ($\ktrust = 42$) &  Alg.~\ref{algo:CP_opt} & $2306$ & \op{}   $\Delta=1$ & $\StaskFour\setminus\{37,53\}$ & $\sim 10^5$ s \\\cline{2-7} %
                                                                                                                                                       & Moderate ($\ktrust = 62$) &  Alg.~\ref{algo:CP} & $2306$ &  \so{}   $\Delta=4.43$ &  \begin{minipage}{3cm}\footnotesize\hspace*{-0.75em} ${\StaskFour\setminus\{37,53\}\cup\{2,52\}}$\end{minipage}& $\sim 10^5$  s \\\cline{2-7} %
                                                                                                                                                       & None ($\ktrust = 108$) &  Alg.~\ref{algo:greedy} & -- & \so{}   $\Delta=4.99$  &  \begin{minipage}{3.1cm} \small $\SensorSet\setminus\{5, 7, 9, 11,15, 17,$ $19, 21, 23, 25,27, 29, 31,$ $33, 35, 37, 39, 41, 43, 45,$ $47, 49, 50, 51, 53, 54\}$ \end{minipage} & $3.23$ s \\\hline
    \end{tabular}}
   \caption{Optimal user-interface solutions and computation time for IEEE 118-bus power grid problem.}
   \label{tab:results}
\end{table*}

The IEEE 118-bus model is a power network composed of $118$
buses, $54$ synchronous machines (generators), $186$
transmission lines, $9$ transformers and $99$
loads~\cite{IEEE118}.  We use linearized swing dynamics to
describe the interconnected generator
dynamics~\cite{van2006usefulness, machowski2011power}.  As
typically done in large networks~\cite{bergen1999power,
dorfler2013kron}, we used Kron reduction to reduce the
network to a generator-only network with LTI
dynamics,
\begin{align}
    \dot{x}(t)&=A_ix(t) + B_ju(t)\label{eq:gen_dyn}.
\end{align}
Here, the state $x(t)\in \mathbb{R}^{108}$ denotes the phase
and phase rate for each of the $54$ generator buses,
and the input $u(t)\in \mathbb{R}^m$ denotes the power
injection provided at each generator bus.  
We construct the system and input matrices,
$A_i\in \mathbb{R}^{108\times108}$ and $B_j\in
\mathbb{R}^{108\times m_j}$, under four
different network configurations:
\begin{enumerate}
    \item Normal operation $(A_1,B_1)$ with
        $m_1=54$,
    \item Load bus $38$ is down $(A_2,B_1)$ with
        $m_1=54$,
    \item Line $65-66$ is down $(A_3,B_1)$ with
        $m_1=54$, and
    \item Alternate generators are down $(A_1,B_2)$ with
        $m_2=27$.
\end{enumerate}
The admittance values for the interconnections in the
reduced network were obtained using
\texttt{MATPOWER}~\cite{zimmerman2011matpower}.  We
considered all the generators to be homogenous. We chose
the moment of inertia and damping coefficients as $H=2.656\
s$ and $D=2$~\cite[Tab.  1]{demetriou2015dynamic}.

We presume the user (power grid operator) is tasked with the
maintaining the power flow to a predetermined substation
(generator bus $28$) under each of the four network configurations.
The power grid operator
therefore requires information about the power flowing from all neighboring
nodes, which can be described as nonlinear functions of the difference in phase
measurements~\cite{van2006usefulness, machowski2011power}.
Therefore, the task is defined in terms of 
the phase measurements of the generator buses that have a
direct connection to the bus $28$ in the Kron reduced
network, consisting of only generator buses.

We define $\SensorSet$ to be all phase measurements of the generators on the Kron reduced network, $\SensorSet=\{e_i: i\in \mathbb{N}_{[1,54]}\}$, where $e_i$ is a column vector of zeros with one at the $i^\mathrm{th}$ component.
For the first three configurations, we have task matrices
$\CtaskOne$, $\CtaskTwo$, and $\CtaskThree$ due to differences
in the neighbors to bus $28$.  Since the network
configuration is the same in the first and the
fourth configurations, the task matrix for the fourth configuration is 
also $\CtaskOne$.

We consider three different trust levels
$\ktrust\in\{\Gamma(\Stask)-10,
\Gamma(\Stask)+10,\Gamma(\SensorSet)\}$.  Informally, this
may be interpreted as designing the user interface under:
\begin{enumerate}
    \item[a)] {\em high trust}: normal operating conditions, in which the user trusts the automation to a high
        degree ($\ktrust=\Gamma(\Stask)-10$), 
    \item[b)] {\em moderate trust}: off-nominal operating conditions, in which the user has some distrust of the
        automation, but not excessive distrust  ($\ktrust=\Gamma(\Stask)+10$), and 
    \item[c)] {\em no trust}: extreme, off-nominal operating conditions, in which the user totally distrusts the
        automation ($\ktrust=\Gamma(\SensorSet)$).
\end{enumerate}
Our results for the four configurations, under each of the three trust levels, is 
shown in 
Table~\ref{tab:results}.

\begin{figure}[ht]
    \centering
    \newcommand{\trimValues}{300 120 250 100}
    \begin{subfigure}[b]{0.6\linewidth}
        \includegraphics[Trim=\trimValues, clip,
        width=\textwidth]{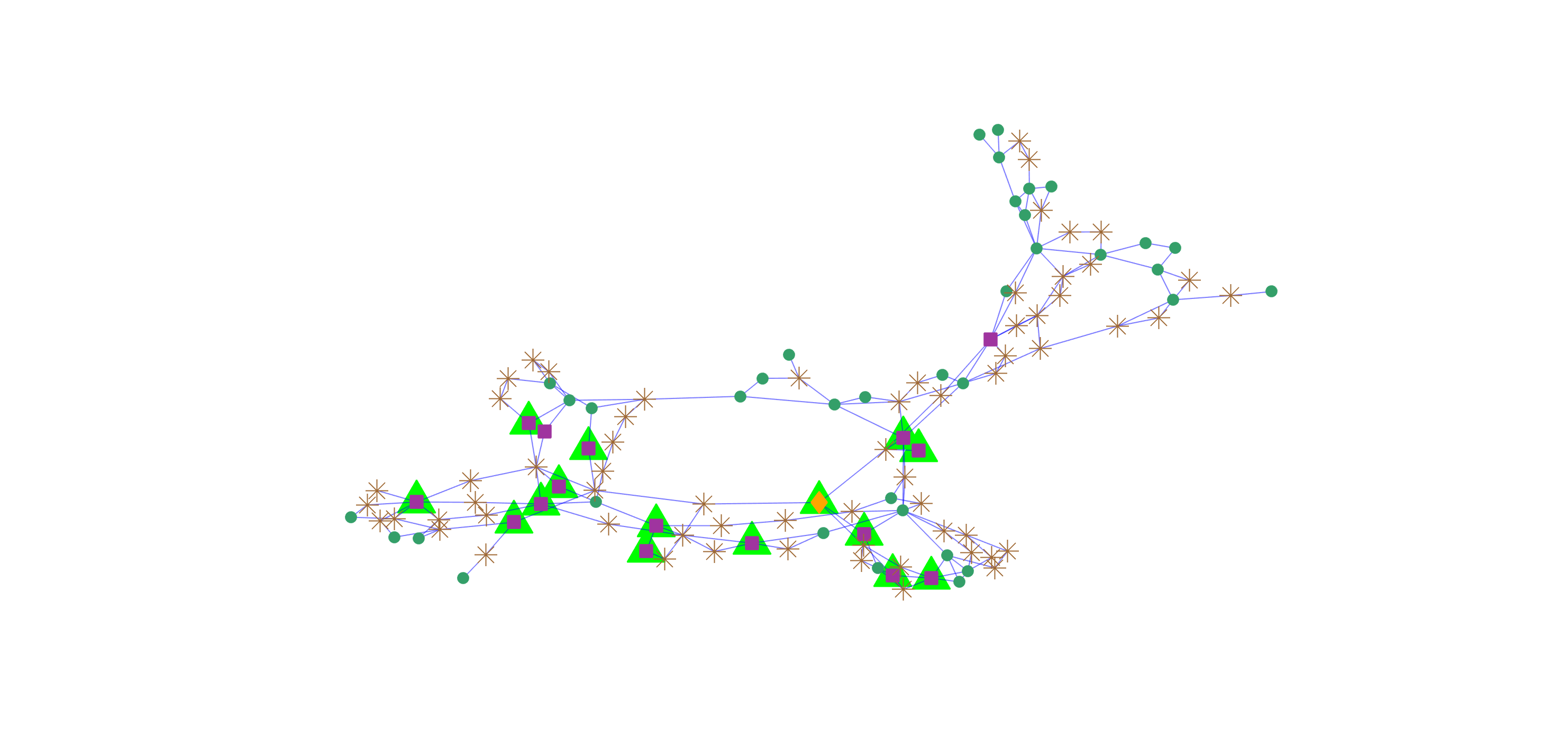}
        \caption{High trust ($\ktrust=42$)}
    \end{subfigure}
    \begin{subfigure}[b]{0.6\linewidth}
        \includegraphics[Trim=\trimValues, clip,
        width=\textwidth]{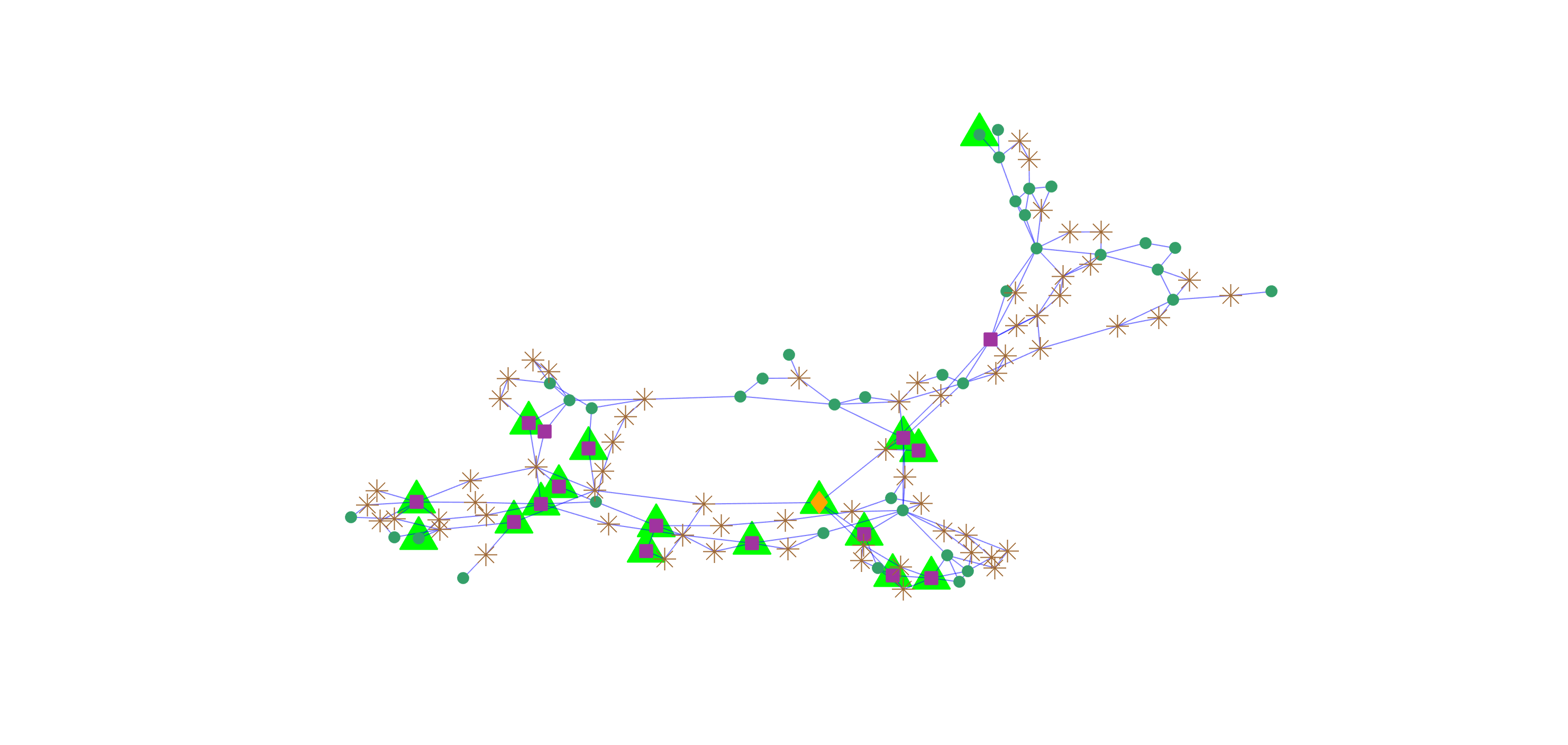}
      \caption{Moderate trust ($\ktrust=62$)}
    \end{subfigure}
    \begin{subfigure}[b]{0.6\linewidth}
        \includegraphics[Trim=\trimValues, clip,
        width=\textwidth]{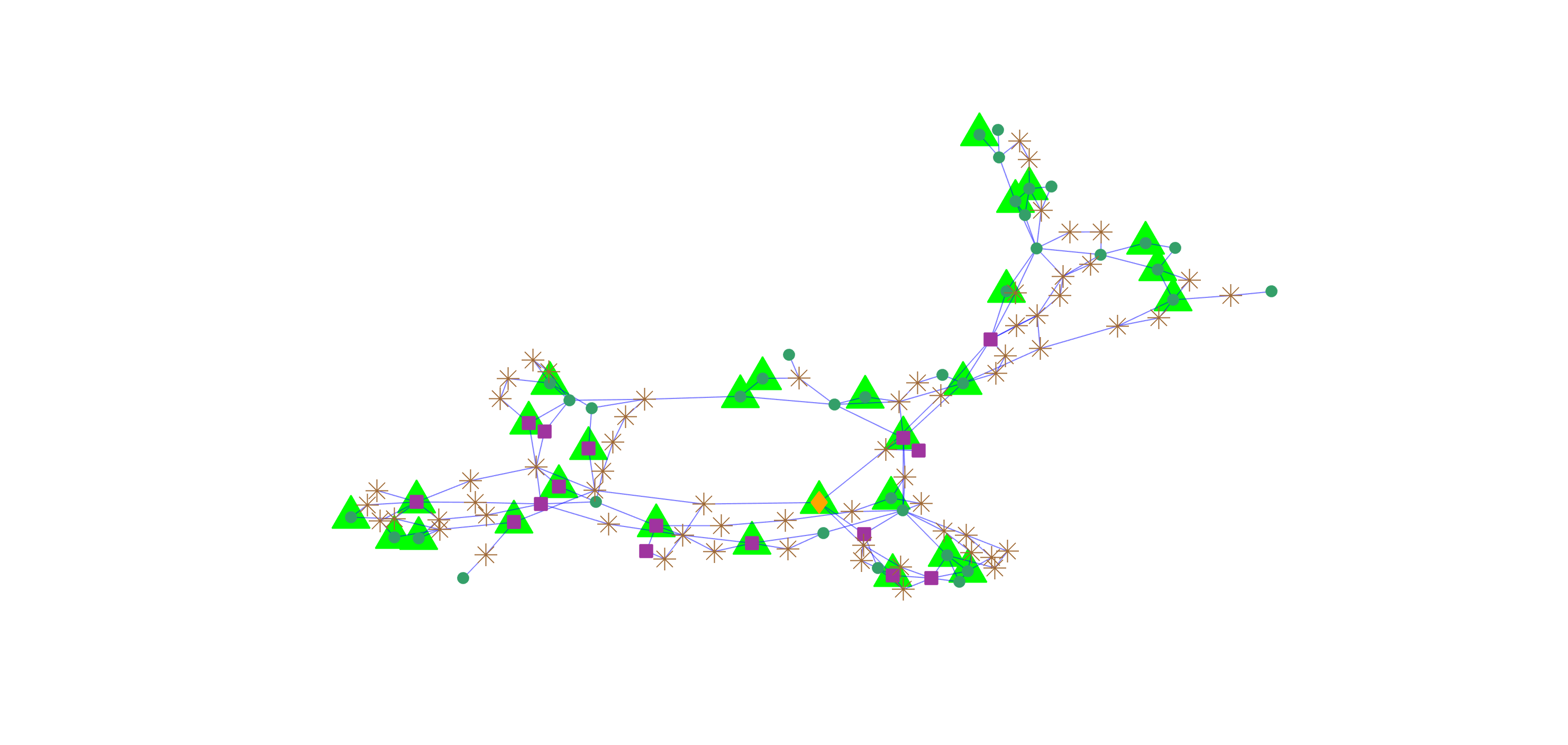}
        \caption{No trust
        ($\ktrust=118$)}
    \end{subfigure}
    \definecolor{mygeneratorbus}{RGB}{52, 159, 105}
    \definecolor{myloadbus}{RGB}{159, 105, 52}
    \definecolor{mybusofinterest}{RGB}{255, 165, 0}
    \definecolor{mytaskcolor}{RGB}{159, 52, 159}
    \definecolor{mysolutioncolor}{RGB}{0, 255, 0}
    \caption{Three interfaces are shown for the IEEE 118-bus
        under Configuration 4, in which alternate generators
        are operational.  The power grid operator's task is
        to maintain power flow at bus 28
        ({\color{mybusofinterest} $\blacklozenge$}).  The user-interface consists of
        selected generator phase angles
        ({\color{mysolutioncolor} $\blacktriangle$});
        neighbouring generators ({\color{mytaskcolor}
            $\blacksquare$}), load buses
            ({\color{myloadbus} {\Large $*$}}) and generator buses
        ({\color{mygeneratorbus} $\bullet$}) are shown for
        clarity. 
As expected, more 
generators must be monitored with lower levels of trust.}
    \label{fig:case_4}
\end{figure}

In configurations 1, 2, and 3, the relative degree $\gamma(s_i)=2$ for
every $s_i\in\SensorSet$. This means that given the phase
measurement of generator bus $i\in \mathbb{N}_{[1,54]}$, the user can
only infer the phase and the phase rate measurement of bus
$i$, but not of the other buses.  
{\em a) High trust:} Due to this decoupling,
$\StaskOne$  is the only user-interface that enables
situation awareness for configurations 1, 2, and 3 under high trust. 
In other words, we need to monitor all the buses
that are involved in the task specification.
{\em b) Moderate trust:} Additional sensors are required; 
due to the decoupling of the generator dynamics in the network, any combination of five previously unselected generators can satisfy the trust constraint.
{\em c) No trust:} Phase measurements from all the buses
would have to be displayed to attain a user information
index of $\Gamma(\SensorSet)=n=108$.  Note that even though
this is the optimal solution to \eqref{prob:UI} for $\ktrust=108$, the
conservative suboptimality bound for
Algorithm~\ref{algo:greedy} is $\Delta=4.99$.

For the fourth configuration (shown in Figure \ref{fig:case_4}), in which 
only alternate generators are operational, %
phase measurements of bus $i$ let the user infer information 
about the network
beyond bus $i$.  
Algorithm~\ref{algo:Enum} returned a non-trivial $\Ssitawarered$, with
$2,306$ elements, each of which 
could enable situation
awareness. {\em a) High trust}: 
Algorithm~\ref{algo:CP_opt} identified 
$\StaskFour\setminus\{37,53\}$ as an optimal user-interface, which in contrast to configurations 1, 2, and 3, provides sensors other than 
those associated with the task.
This interface exploits the underlying dynamics and the user's situation awareness, so that phases of the task-relevant generators can be reconstructed based on less information than would be provided merely by duplicating the sensors associated with the task. Specifically, paths between nodes 28, 37, and 53 in the network topology of the 118-bus grid, which appear in a block of the dynamics matrix in the generator swing equations, allow the user to reconstruct relevant states with fewer sensors at this level of trust.
{\em b) Moderate trust:} %
Algorithm~\ref{algo:CP} determined that two additional sensors were required. Since the user has some distrust in the automation in this case, the additional sensors reveal additional paths in the network topology that can be observed through the dynamics, which allows the user to reconstruct supplemental states and their derivatives relevant to monitoring bus 28 to meet the trust constraint. 
{\em c) No trust:} %
Algorithm~\ref{algo:greedy} yields a set of $28$ sensors that are needed to monitor power flow to bus 28. When the user fully distrusts the automation, a relatively large number of sensors are required to allows the user to reconstruct and understand states relevant to monitoring bus 28 that previously were entrusted to the automation. These sensors correspond to generators spread across the network, in order to improve observability over the entire system.

Using efficient enumeration techniques described in
Section~\ref{sec:CP}, enumerating $2^\Sred$ with $| 2^{22}| \approx
4\times10^6$ elements took only $\sim 10^5$ seconds
(about $27$ hours) to
compute.  In contrast, a naive approach using linear search over
$\PowerSet$ for the minimum cardinality set that satisfies
the constraints would require checking
$|\PowerSet|=2^{108}\approx 3\times10^{32}$ elements, 
resulting in approximately $10^{13}$ billion hours computation time (presuming each evaluation takes $10^{-4}$ seconds).
All computations were performed using MATLAB on an Intel i7-4600U CPU with 4
cores, 2.1GHz clock rate and 7.5 GB RAM.

\section{Conclusion}
\label{sec:conc}

This paper presents a method for user-interface design via sensor selection.  Unlike many UI based approaches, our method is driven by the underlying dynamics of the human-automation system, and constrained by user's situation awareness and trust in the automation.
We use submodular maximization and constraint programming to solve the sensor selection problem as a constrained combinatorial optimization, and exploit submodularity and monotonicity properties to 
identify optimal or sub-optimal solutions.
We applied our approach to a large human-automation system, consisting of a power grid operator for the IEEE 118-bus model, and constructed correct-by-design interfaces for a variety of trust levels and operating scenarios.

\section*{Acknowledgements}

The authors thank Hasan Poonawala for
his input in the formulation of the combinatorial
optimization problem \eqref{prob:feas_soneset}.


\appendix

\subsection{Submodularity in combinatorial optimization problems}
\label{app:sub_max}

Let $\SensorSet$ denote a finite set.  Consider the
following combinatorial optimization problem with a
submodular, monotone increasing set function $f: \PowerSet
\rightarrow \mathbb{N}_{[0, f(\SensorSet)]}$,
\begin{align}
    \begin{array}{rl}
        \underset{\mcS\in\PowerSet}{\mathrm{minimize}}& | \mcS| \\
        \mathrm{subject\ to}& f(\mcS)\geq k\\
    \end{array}\label{prob:subm_max}
\end{align}
for some problem parameter $k\in\mathbb{N}_{[1, f(\SensorSet)]}$.

\begin{lem}[\textbf{Suboptimality bound for the greedy
    solution to
\eqref{prob:subm_max}~\cite{clark_submodularity_2017,wolsey1982analysis}}]\label{lem:greedy}
    Submodular maximization problem \eqref{prob:subm_max} admits a $ \mathcal{O}({| \SensorSet|}^2)$ greedy algorithm (Algorithm~\ref{algo:greedy}) such that its solution $\Sgreedy$ satisfies the property
    \begin{align}
        1\leq \frac{| \Sgreedy|}{| \Sopt |}&\leq 1+\log\left(\frac{f(\SensorSet) -
        f(\emptyset)}{f(\Sgreedy)-f(\SgreedyMinus)}\right)\label{eq:subopt_bound}
    \end{align}
    with $\SgreedyMinus$ is the solution at the iteration
    prior to termination of Algorithm~\ref{algo:greedy}. 
\end{lem}

\begin{algorithm}[t]
    \caption{Greedy algorithm to solve \eqref{prob:subm_max}}\label{algo:greedy}
    \begin{algorithmic}[1]
        \Require~Submodular monotone increasing function $f(\cdot)$, power set $\PowerSet$, submodular function lower bound $k$
        \Ensure~Optimal greedy solution $\Sgreedy$ and the solution prior to termination step $\SgreedyMinus$ (see \eqref{eq:subopt_bound})
        \Procedure{GreedyAlgorithm}{}%
            \State $\Sgreedy\gets\emptyset$
            \While{$f(\Sgreedy)< k$}
                \State $s^\ast\gets \underset{s\in \SensorSet\setminus\Sgreedy}{\operatorname{argmax}} f(\Sgreedy \cup \{s\})-f(\Sgreedy)$ %
                \State $\SgreedyMinus \gets \Sgreedy$, $\Sgreedy \gets \Sgreedy \cup \{s^\ast\}$ %
            \EndWhile
            \State \Return $\Sgreedy,\SgreedyMinus$
      \EndProcedure
  \end{algorithmic}
\end{algorithm}

Algorithm~\ref{algo:greedy} is a greedy approach to solve the submodular optimization problem
\eqref{prob:subm_max} with provable worst-case suboptimality
bounds (Lemma~\ref{lem:greedy}).  The suboptimality bound
given by
Lemma~\ref{lem:greedy} is the best bound available by any polynomial-time
algorithm~\cite{feige1998threshold}, assuming $P\neq NP$.  The bound in
\eqref{eq:subopt_bound} is a worst-case bound; Algorithm~\ref{algo:greedy} often
performs significantly better in practice~\cite{clark_submodularity_2017}.

\subsection{Proof of Lemma~\ref{lem:prop_TS_gamma}}
\label{app:lem_prop_TS_gamma}
\proofLemPropTSGamma{}
\qed

\subsection{Proof of Thm.~\ref{thm:Ssitaware_decompose}}
\label{app:thm_ssitware}
\proofThmSsitaware{}
\qed

\subsection{Proof of Lemma~\ref{lem:Stask_feas}}
\label{app:lem_Stask_feas}
\proofLemStaskFeas{}
\qed

\subsection{Proof of
Corollary~\ref{corr:lower_bound_mcS}}
\proofCorrLowerBoundMcS{}
\qed

\subsection{Proof of Corollary~\ref{corr:Strust}}
\proofCorrStrust{}
\qed


\begin{thebibliography}{10}
\providecommand{\url}[1]{#1}
\csname url@samestyle\endcsname
\providecommand{\newblock}{\relax}
\providecommand{\bibinfo}[2]{#2}
\providecommand{\BIBentrySTDinterwordspacing}{\spaceskip=0pt\relax}
\providecommand{\BIBentryALTinterwordstretchfactor}{4}
\providecommand{\BIBentryALTinterwordspacing}{\spaceskip=\fontdimen2\font plus
\BIBentryALTinterwordstretchfactor\fontdimen3\font minus
  \fontdimen4\font\relax}
\providecommand{\BIBforeignlanguage}[2]{{%
\expandafter\ifx\csname l@#1\endcsname\relax
\typeout{** WARNING: IEEEtran.bst: No hyphenation pattern has been}%
\typeout{** loaded for the language `#1'. Using the pattern for}%
\typeout{** the default language instead.}%
\else
\language=\csname l@#1\endcsname
\fi
#2}}
\providecommand{\BIBdecl}{\relax}
\BIBdecl

\bibitem{Endsley1995}
M.~Endsley, ``Toward a theory of situation awareness in dynamic systems,''
  \emph{Human Factors}, vol.~37, pp. 32--64, 1995.

\bibitem{sheridan1992telerobotics}
T.~Sheridan, \emph{Telerobotics, automation, and human supervisory
  control}.\hskip 1em plus 0.5em minus 0.4em\relax MIT press, 1992.

\bibitem{Dix91}
A.~Dix, \emph{Formal Methods for Interactive Systems}.\hskip 1em plus 0.5em
  minus 0.4em\relax Academic Press, 1991.

\bibitem{Panteli2013}
M.~{Panteli}, P.~{Crossley}, D.~{Kirschen}, and D.~{Sobajic}, ``Assessing the
  impact of insufficient situation awareness on power system operation,''
  \emph{IEEE Trans. Power Syst.}, vol.~28, no.~3, pp. 2967--2977, 2013.

\bibitem{Endsley2008}
M.~Endsley and E.~Connors, ``Situation awareness: State of the art,'' in
  \emph{IEEE Power and Energy Society General Meeting-Conversion and Delivery
  of Electrical Energy in the 21st Century}, 2008, pp. 1--4.

\bibitem{steinfeld2004interface}
A.~Steinfeld, ``Interface lessons for fully and semi-autonomous mobile
  robots,'' in \emph{Proc. IEEE Int'l Conf. Robotics \& Autom.}, vol.~3, 2004,
  pp. 2752--2757.

\bibitem{PF11}
A.~Pritchett and M.~Feary, ``Designing human-automation interaction,''
  \emph{Handbook of Human-Machine Interaction}, pp. 267--282, 2011.

\bibitem{Billings97}
C.~Billings, \emph{Aviation Automation: The Search for a Human-Centered
  Approach}.\hskip 1em plus 0.5em minus 0.4em\relax Erlbaum, 1997.

\bibitem{DH02}
A.~Degani and M.~Heymann, ``Formal verification of human-automation
  interaction,'' \emph{Human Factors}, vol.~44, no.~1, pp. 28--43, 2002.

\bibitem{VT_RSS_2010}
M.~Vitus and C.~Tomlin, ``Sensor placement for improved robotic navigation,''
  in \emph{Proc. Robotics: Science and Syst.}, 2010.

\bibitem{mourikis2006optimal}
A.~Mourikis and S.~Roumeliotis, ``Optimal sensor scheduling for
  resource-constrained localization of mobile robot formations,'' \emph{IEEE
  Trans. Robotics}, vol.~22, no.~5, pp. 917--931, 2006.

\bibitem{rowaihy2007survey}
H.~Rowaihy, S.~Eswaran, M.~Johnson, D.~Verma, A.~Bar-Noy, T.~Brown, and
  T.~La~Porta, ``A survey of sensor selection schemes in wireless sensor
  networks,'' in \emph{Defense and Security Symposium}, 2007.

\bibitem{qi2014optimal}
J.~Qi, K.~Sun, and W.~Kang, ``Optimal {PMU} placement for power system dynamic
  state estimation by using empirical observability {Gramian},'' \emph{IEEE
  Trans. Power Syst.}, vol.~30, pp. 2041--2054, 2015.

\bibitem{gou2008generalized}
B.~Gou, ``Generalized integer linear programming formulation for optimal {PMU}
  placement,'' \emph{IEEE Trans. Power Syst.}, vol.~23, no.~3, pp. 1099--1104,
  2008.

\bibitem{JoshiConvex}
S.~Joshi and S.~Boyd, ``Sensor selection via convex optimization,'' \emph{IEEE
  Trans. Signal Processing}, vol.~57, no.~2, pp. 451--462, 2009.

\bibitem{polyak2013lmi}
B.~Polyak, M.~Khlebnikov, and P.~Shcherbakov, ``An {LMI} approach to structured
  sparse feedback design in linear control systems,'' in \emph{Proc. Euro.
  Ctrl. Conf.}, 2013, pp. 833--838.

\bibitem{krause2007near}
A.~Krause and C.~Guestrin, ``Near-optimal observation selection using
  submodular functions,'' in \emph{AAAI}, vol.~7, 2007, pp. 1650--1654.

\bibitem{shamaiah2010greedy}
M.~Shamaiah, S.~Banerjee, and H.~Vikalo, ``Greedy sensor selection: Leveraging
  submodularity,'' in \emph{Proc. IEEE Conf. Dec. \& Ctrl.}, 2010, pp.
  2572--2577.

\bibitem{summers_2014}
T.~Summers, F.~L. Cortesi, and J.~Lygeros, ``On submodularity and
  controllability in complex dynamical networks,'' \emph{IEEE Trans. Control of
  Network Syst.}, vol.~3, no.~1, pp. 1--11, 2016.

\bibitem{nemhauser_1978}
G.~Nemhauser, L.~Wolsey, and M.~Fisher, ``An analysis of approximations for
  maximizing submodular set functions - \uppercase{I},'' \emph{Mathematical
  Programming}, vol.~14, no.~1, pp. 265--294, 1978.

\bibitem{clark_submodularity_2017}
A.~Clark, B.~Alomair, L.~Bushnell, and R.~Poovendran, ``Submodularity in
  {Input} {Node} {Selection} for {Networked} {Linear} {Systems}: {Efficient}
  {Algorithms} for {Performance} and {Controllability},'' \emph{IEEE Control
  Syst.}, vol.~37, no.~6, pp. 52--74, Dec. 2017.

\bibitem{wolsey1982analysis}
L.~Wolsey, ``An analysis of the greedy algorithm for the submodular set
  covering problem,'' \emph{Combinatorica}, vol.~2, no.~4, pp. 385--393, 1982.

\bibitem{berger2005discrete}
T.~Berger-Wolf, W.~Hart, and J.~Saia, ``Discrete sensor placement problems in
  distribution networks,'' \emph{Mathematical and Computer Modelling}, vol.~42,
  no.~13, pp. 1385--1396, 2005.

\bibitem{Rushby14}
G.~Gelman, K.~Feigh, and J.~Rushby, ``Example of a complementary use of model
  checking and human performance simulation,'' \emph{IEEE Trans. Human-Machine
  Syst.}, vol.~44, no.~5, pp. 576--590, October 2014.

\bibitem{BBS13}
M.~Bolton, E.~Bass, and R.~Siminiceanu, ``Using formal verification to evaluate
  human-automation interaction: A review,'' \emph{{IEEE} Trans. Syst., Man, and
  Cybernetics: Syst.}, vol.~43, no.~3, pp. 488--503, 2013.

\bibitem{OMBDT08}
M.~Oishi, I.~Mitchell, A.~Bayen, and C.~Tomlin, ``Invariance-preserving
  abstractions of hybrid systems: {Application} to user interface design,''
  \emph{{IEEE} Trans. Ctrl. Syst. Tech.}, vol.~16, no.~2, pp. 229--244, March
  2008.

\bibitem{Sadigh2017}
T.~{Rezvani}, K.~{Driggs-Campbell}, D.~{Sadigh}, S.~{Sastry}, S.~{Seshia}, and
  R.~{Bajcsy}, ``Towards trustworthy automation: User interfaces that convey
  internal and external awareness,'' in \emph{Proc. IEEE Int'l Conf. on
  Intelligent Transportation Syst.}, Nov 2016, pp. 682--688.

\bibitem{Wang2018}
L.~Jiang and Y.~Wang, ``A human-computer interface design for quantitative
  measure of regret theory,'' \emph{IFAC-PapersOnLine}, vol.~51, no.~34, pp. 15
  -- 20, 2019.

\bibitem{Jain2018_1}
K.~Akash, K.~Polson, T.~Reid, and N.~Jain, ``Improving human-machine
  collaboration through transparency-based feedback – part {I}: Human trust
  and workload model,'' \emph{IFAC-PapersOnLine}, vol.~51, no.~34, pp. 315 --
  321, 2019.

\bibitem{Jain2018_2}
K.~Akash, T.~Reid, and N.~Jain, ``Improving human-machine collaboration through
  transparency-based feedback – part {II}: Control design and synthesis,''
  \emph{IFAC-PapersOnLine}, vol.~51, no.~34, pp. 322 -- 328, 2019.

\bibitem{BLM12}
J.~Bailleiul, N.~Leonard, and K.~Morgansen, ``Interaction dynamics: The
  interface of humans and smart machines,'' \emph{Proc. {IEEE}}, vol. 100,
  no.~3, pp. 567--570, 2012.

\bibitem{Murphey2016}
K.~{Fitzsimons}, E.~{Tzorakoleftherakis}, and T.~{Murphey}, ``Optimal
  human-in-the-loop interfaces based on {M}axwell's {D}emon,'' in \emph{Proc.
  Amer. Ctrl. Conf.}, 2016, pp. 4397--4402.

\bibitem{Sadrfaridpour2017}
B.~Sadrfaridpour, H.~Saeidi, J.~Burke, K.~Madathil, and Y.~Wang, ``Modeling and
  control of trust in human-robot collaborative manufacturing,'' in
  \emph{Robust Intelligence and Trust in Autonomous Syst.}, 2016, pp. 115--141.

\bibitem{Saeidi2017}
H.~{Saeidi}, J.~{Wagner}, and Y.~{Wang}, ``A mixed-initiative haptic
  teleoperation strategy for mobile robotic systems based on bidirectional
  computational trust analysis,'' \emph{IEEE Trans. Robotics}, vol.~33, no.~6,
  pp. 1500--1507, 2017.

\bibitem{EO11}
N.~Eskandari and M.~Oishi, ``Computing observable and predictable subspaces to
  evaluate user-interfaces of {LTI} systems under shared control,'' in
  \emph{IEEE Conf. Syst., Man and Cybernetics}, 2011, pp. 2803--2808.

\bibitem{oishi_2014}
M.~Oishi, ``Assessing information availability for user-interfaces of shared
  control systems under reference tracking,'' in \emph{Proc. Amer. Ctrl.
  Conf.}, 2014, pp. 3474--3481.

\bibitem{Hammond15}
T.~Hammond, N.~Eskandari, and M.~Oishi, ``Observability of user-interfaces for
  hybrid {LTI} systems under collaborative control: Application to aircraft
  flight management systems,'' \emph{{IEEE} Trans. Automation Science and
  Engineering}, vol.~13, no.~1, pp. 78--84, 2015.

\bibitem{VinodACC2016}
A.~Vinod, T.~Summers, and M.~Oishi, ``User-interface design for {MIMO LTI}
  human-automation systems through sensor placement,'' in \emph{Proc. Amer.
  Ctrl. Conf.}, 2016, pp. 5276--5283.

\bibitem{friedbergLinearAlg}
S.~Friedberg, A.~Insel, and L.~Spence, \emph{Linear algebra}, 4th~ed.\hskip 1em
  plus 0.5em minus 0.4em\relax Pearson Higher Ed, 2014.

\bibitem{baier2008principles}
C.~Baier and J.-P. Katoen, \emph{Principles of model checking}.\hskip 1em plus
  0.5em minus 0.4em\relax MIT press, 2008.

\bibitem{Eskandari2016}
N.~{Eskandari}, G.~{Dumont}, and Z.~{Wang}, ``An observer/predictor-based model
  of the user for attaining situation awareness,'' \emph{IEEE Trans. on
  Human-Machine Syst.}, vol.~46, no.~2, pp. 279--290, 2016.

\bibitem{Oulasvirta2020}
A.~{Oulasvirta}, N.~{Dayama}, M.~{Shiripour}, M.~{John}, and A.~{Karrenbauer},
  ``Combinatorial optimization of graphical user interface designs,''
  \emph{Proc. IEEE}, vol. 108, no.~3, pp. 434--464, 2020.

\bibitem{wickens2008situation}
C.~Wickens, ``Situation {A}wareness: {R}eview of {M}ica {E}ndsley's 1995
  articles on {S}ituation {A}wareness {T}heory and {M}easurement,'' \emph{Human
  Factors: The Journal of the Human Factors and Ergonomics Society}, vol.~50,
  no.~3, pp. 397--403, 2008.

\bibitem{Hoff15}
K.~A. Hoff and M.~Bashir, ``Trust in automation: Integrating empirical evidence
  on factors that influence trust,'' \emph{Human Factors}, vol.~57, no.~3, pp.
  407--434, 2015.

\bibitem{Jain2019}
W.~{Hu}, K.~{Akash}, T.~{Reid}, and N.~{Jain}, ``Computational modeling of the
  dynamics of human trust during human–machine interactions,'' \emph{IEEE
  Trans. on Human-Machine Syst.}, vol.~49, no.~6, pp. 485--497, 2019.

\bibitem{ParasuramanRiley97}
R.~Parasuraman and V.~Riley, ``Humans and automation: Use, misuse, disuse,
  abuse,'' \emph{Human Factors}, vol.~39, no.~2, pp. 230--253, 1997.

\bibitem{LeeSee04}
J.~D. Lee and K.~A. See, ``Trust in automation: Designing for appropriate
  reliance,'' \emph{Human Factors}, vol.~46, no.~1, pp. 50--80, 2004.

\bibitem{PSW08}
R.~Parasuraman, T.~Sheridan, and C.~Wickens, ``Situation awareness, mental
  workload, and trust in automation: Viable, empirically supported cognitive
  engineering constructs,'' \emph{J. Cognitive Engineering and Decision
  Making}, vol.~2, no.~2, pp. 140--160, 2008.

\bibitem{LewisFTA2018}
M.~Lewis, K.~Sycara, and P.~Walker, \emph{The Role of Trust in Human-Robot
  Interaction}.\hskip 1em plus 0.5em minus 0.4em\relax Springer Int'l
  Publishing, 2018, pp. 135--159.

\bibitem{SATI}
P.~Goillau, C.~Kelly, M.~Boardman, and E.~Jeannot, ``Guidelines for trust in
  future {ATM} systems: Measures,'' \emph{{EUROCONTROL}, the European
  Organisation for the Safety of Air Navigation}, 2003.

\bibitem{HTC}
M.~Madsen and S.~Gregor, ``Measuring human-computer trust,'' in \emph{11th
  Australasian Conf. Information Syst.}, vol.~53.\hskip 1em plus 0.5em minus
  0.4em\relax Citeseer, 2000, pp. 6--8.

\bibitem{ED}
J.-Y. Jian, A.~Bisantz, and C.~Drury, ``Foundations for an empirically
  determined scale of trust in automated systems,'' \emph{Int'l J. Cognitive
  Ergonomics}, vol.~4, no.~1, pp. 53--71, 2000.

\bibitem{setter2016trust}
T.~Setter, A.~Gasparri, and M.~Egerstedt, ``Trust-based interactions in teams
  of mobile agents,'' in \emph{Proc. Amer. Ctrl. Conf.}\hskip 1em plus 0.5em
  minus 0.4em\relax IEEE, 2016, pp. 6158--6163.

\bibitem{fujishige2005submodular}
S.~Fujishige, \emph{Submodular functions and optimization}.\hskip 1em plus
  0.5em minus 0.4em\relax Elsevier, 2005, vol.~58.

\bibitem{Lovasz1983}
L.~Lov{\'a}sz, ``Submodular functions and convexity,'' in \emph{Mathematical
  Programming The State of the Art}.\hskip 1em plus 0.5em minus 0.4em\relax
  Springer, 1983, pp. 235--257.

\bibitem{krause2014submodular}
A.~Krause and D.~Golovin, ``Submodular function maximization,'' in
  \emph{Tractability: Practical Approaches to Hard Problems}.\hskip 1em plus
  0.5em minus 0.4em\relax Cambridge University Press, 2014, pp. 71--104.

\bibitem{rossi2006handbook}
F.~Rossi, P.~Van~Beek, and T.~Walsh, \emph{Handbook of constraint
  programming}.\hskip 1em plus 0.5em minus 0.4em\relax Elsevier, 2006.

\bibitem{IEEE118}
\BIBentryALTinterwordspacing
``{IEEE}-$118$ power network.'' [Online]. Available:
  \url{http://www.ee.washington.edu/research/pstca/pf118/pg_tca118bus.htm}
\BIBentrySTDinterwordspacing

\bibitem{van2006usefulness}
D.~Van~Hertem, J.~Verboomen, K.~Purchala, R.~Belmans, and W.~Kling,
  ``Usefulness of {DC} power flow for active power flow analysis with flow
  controlling devices,'' in \emph{Intn'l Conf. on AC and DC Power
  Transmission}.\hskip 1em plus 0.5em minus 0.4em\relax IET, 2006, pp. 58--62.

\bibitem{machowski2011power}
J.~Machowski, J.~Bialek, and J.~Bumby, \emph{Power system dynamics: stability
  and control}.\hskip 1em plus 0.5em minus 0.4em\relax John Wiley \& Sons,
  2011.

\bibitem{bergen1999power}
A.~Bergen and V.~Vittal, \emph{Power System Analysis}, 2nd~ed.\hskip 1em plus
  0.5em minus 0.4em\relax Prentice Hall, 1999.

\bibitem{dorfler2013kron}
F.~Dorfler and F.~Bullo, ``Kron reduction of graphs with applications to
  electrical networks,'' \emph{IEEE Trans. Circ. Syst.}, vol.~60, no.~1, pp.
  150--163, 2013.

\bibitem{zimmerman2011matpower}
R.~Zimmerman, C.~Murillo-S{\'a}nchez, and R.~Thomas, ``{MATPOWER}: Steady-state
  operations, planning, and analysis tools for power systems research and
  education,'' \emph{IEEE Trans. Power Syst.}, vol.~26, no.~1, pp. 12--19,
  2011.

\bibitem{demetriou2015dynamic}
P.~Demetriou, M.~Asprou, J.~Quiros-Tortos, and E.~Kyriakides, ``Dynamic {IEEE}
  test systems for transient analysis,'' \emph{IEEE Syst. Journal}, vol.~11,
  pp. 2108--2117, 2017.

\bibitem{feige1998threshold}
U.~Feige, ``A threshold of ln n for approximating set cover,'' \emph{J. ACM},
  vol.~45, no.~4, pp. 634--652, 1998.

\end{thebibliography}
\end{document}